\documentclass[9pt]{revtex4}
\usepackage{amsmath}
\usepackage{amsfonts}
\usepackage{amssymb}
\usepackage{mathrsfs}
\usepackage{graphicx}
\usepackage{subcaption}
\usepackage{multirow}
\usepackage{url}
\usepackage{color}
\newtheorem{proposition}{PROPOSITION}
\newenvironment{proof}{{\noindent\it Proof:}\quad}{\hfill $\square$\par}  
\def\Q{\mathbf{Q}}
\def\P{\mathbf{P}}
\def\I{\mathbf{I}}
\def\n{\mathbf{n}}
\def\r{\mathbf{r}}
\def\x{\mathbf{x}}
\def\y{\mathbf{y}}
\def\z{\mathbf{z}}

\begin{document}
\title{A Reduced Study for Nematic Equilibria on Two-Dimensional Polygons}
\author{Yucen Han$^{1}$}
\author{Apala Majumdar$^{2}$}
\author{Lei Zhang$^{3}$}
\affiliation{$^1$Beijing International Center for Mathematical Research, Peking University, Beijing 100871, China.\\
$^2$Department of Mathematics and Statistics, University of Strathclyde, Glasgow, G1 1XH, United Kingdom and Department of Mathematics Science, University of Bath, United Kingdom.\\
$^3$Beijing International Center for Mathematical Research, Center for Quantitative Biology, Peking University, Beijing 100871, China.}

\begin{abstract}
 We study reduced nematic equilibria on regular two-dimensional polygons with Dirichlet tangent boundary conditions, in  a reduced two-dimensional Landau-de Gennes framework, discussing their relevance in the full three-dimensional framework too. We work at a fixed temperature and study the reduced stable equilibria in terms of the edge length, $\lambda$ of the regular polygon, $E_K$ with $K$ edges. We analytically compute a novel "ring solution" in the $\lambda \to 0$ limit, with a unique point defect at the centre of the polygon for $K \neq 4$. The ring solution is unique. For sufficiently large $\lambda$, we deduce the existence of at least $\left[K/2 \right]$ classes of stable equilibria and numerically compute bifurcation diagrams for reduced equilibria on a pentagon and hexagon, as a function of $\lambda^2$, thus illustrating the effects of geometry on the structure, locations and dimensionality of defects in this framework.
\end{abstract}

\maketitle
\section{Introduction}
Nematic liquid crystals (NLCs) are paradigm examples of soft orientationally ordered materials intermediate between solid and liquid phases of matter, with a degree of long-range orientational order. The orientational order manifests as distinguished directions of molecular alignment leading to anisotropic mechanical, optical and rheological properties \cite{dg, stewart}. NLCs are best known for their applications in the thriving liquid crystal display industry \cite{majumdar_pre_2007, kitson_geisow_apl} but they have tremendous potential in nanoscience, biophysics and materials design, all of which rely on a systematic theoretical approach to the study of NLC equilibria and dynamics. Further, these theoretical approaches promise a suite of technical tools for related applications in the study of surface/interfacial phenomena, active matter, polymers, elastomers and colloid science \cite{Marcus2012Transition,Cai2017Liquid,Teramoto2010Morphological,han2020pathways} and hence, have purpose beyond the specific field of NLCs.

This paper focuses on certain specific questions about stable NLC textures in two-dimensional (2D) domains and these questions are within the broad remit of pattern formation in partially ordered media in confinement, with emphasis on the effects of geometry and boundary conditions without any external fields. The set-up is simple but can give excellent insight into the energetic and geometric origins of interior and boundary defects, stable and unstable patterns and deeper questions pertaining to how we can tune stability by tuning defects, how do we classify unstable states, the role of unstable states in the energy landscape and in the longer-term, how does a system select an unstable transient state during switching mechanisms between distinct stable NLC equilibria. These are fundamental theoretical questions at the interface of topology, analysis, modelling and scientific computation with deep-rooted implications for physics and materials engineering. In particular, with sweeping experimental advances in designing micropatterned surfaces, thin three-dimensional (3D) geometries and 3D printing \cite{Igor2006Two,Mu2008Self}, 2D studies are of practical value. In Section \ref{sec:preliminaries}, we review the reduced Landau-de Gennes approach for modelling nematic liquid crystals (see \cite{kralj2014order} and  \cite{bisht_epl}), which has been used with success to describe the in-plane NLC profiles in 2D domains or thin 3D geometries. This approach assumes that the important structural details can be described by a 2D approach, and the structural details are invariant along the height of the thin 3D domain. As will be discussed below, these 2D predictions may also survive in 3D scenarios. For example,
    in \cite{han2019trans}, the planar radial and planar polar solutions in a 2D disc can also be extended to a 3D cylinder with z-invariance and in \cite{canevari_majumdar_wang_harris}, the authors show that the 2D WORS (Well Order Reconstruction Solution) also exists in a 3D well with a square cross-section. Of course, the 3D scenario is much richer and cannot be exhaustively described by a reduced 2D approach.
In Section~\ref{sec:distinguished_limit}, we study the stable nematic equilibria for a reduced 2D problem on a regular polygon $E_K$ with $K$ edges, in terms of the edge length, $\lambda$, of the polygon, keeping all other parameters fixed in the study. We first study the $\lambda \to 0$ limit for which the reduced problem is a Dirichlet boundary value problem for the Laplace equation on a regular polygon. We use the Schwarz-Christoffel mapping to map a disc to a polygon, solve the corresponding boundary-value problem on a disc, study the limiting unique solution and its rotation/reflection symmetries analytically and label the limiting profile as the new \emph{Ring} solution, which depends on the number of edges, $K$, of a regular polygon $E_K$.  In this limit, we can accurately capture the structure and location of the optical defect, which is mathematically identified with the zero set of the reduced solution.

The optical defect of the ring solution has the profile of a $-1/2$ defect for a triangle, is a pair of mutually orthogonal lines for a square and has the profile of a $+1$-degree Ginzburg-Landau vortex for $K>4$.  In Section~\ref{subsec:infty_limit}, we present some heuristics for the number of stable reduced equilibria in the $\lambda \to \infty$ limit (analogous to Type II superconductors in the GL theory); a simple estimate shows that there are at least $K\choose 2$ stable states which can be analytically computed by solving an associated boundary-value problem for a scalar function.

In Section~\ref{subsec:numerics}, we use both sets of analytic results to compute initial conditions for numerical solvers and use continuation methods to compute bifurcation diagrams for the reduced equilibria on a pentagon and  a hexagon, as illustrative examples. These two examples highlight certain generic differences between polygons with even and odd numbers of sides. As K increases, we have at least [K/2] classes of stable equilibria, distinguished by the locations of a pair of fractional point defects. Each point defect is either pinned at or near a polygon vertex and the different stable states are generated by different defect locations. We do not have good estimates for the number of unstable states, but we do find BD solutions (see \cite{wang2019order} for the origin of the name) in the cases of a pentagon and hexagon, which are unstable equilibria with approximate interior line defects or interior lines of low order. Numerically, when $\lambda$ is small, the BD solutions are index 1 saddle points of the reduced LdG energy that can connect stable equilibria. Whilst our numerical studies are not exhaustive, it is clear that the unstable states are also generated by the symmetries of the polygons and we can build a hierarchy of unstable states and their unstable directions by exploiting the geometry of the problem. As $K \to \infty$, the number of stable states increases rapidly but the stability is closely connected to the curvature of the boundary. For a completely smooth boundary e.g. disc, we lose the rich solution landscape of $E_K$ with $K$ large. In fact, for a disc, in the $R\to \infty$ limit of large radius, we only have the planar polar equilibria featured by two interior nematic point defects along a disc diameter \cite{hu2016disclination, han2019trans}, for appropriately defined boundary conditions. The number of edges, the length of the polygon edge and the sharpness of the polygon vertices give us a diverse set of stable equilibria profiles and precise control on the number and location of defects for new experimental and theoretical studies. We present our conclusions in Section~\ref{sec:conclusion}.

\section{Theoretical Framework}
\label{sec:preliminaries}
The LdG theory is a powerful continuum theory for nematic liquid crystals and
describes the nematic state by a macroscopic order parameter--the LdG $\Q$-tensor,
which is a measure of nematic orientational order. Mathematically, the $\Q$-tensor is a symmetric traceless $3\times$3 matrix i.e.
\begin{equation}
    \Q\in S_0:=\{\Q\in \mathbb{M}^{3\times 3}: Q_{ij} = Q_{ji},Q_{ii} = 0\}\nonumber
\end{equation}
A $\Q$-tensor is said to be (i) isotropic if~$\Q=0$, (ii) uniaxial if $\Q$
has a pair of degenerate non-zero eigenvalues and (iii) biaxial if~$\Q$ has three distinct eigenvalues~\cite{dg}. A
uniaxial $\Q$-tensor can be written in terms of its ``order parameter" and ``director" as follows -  $\Q_u = s \left(\n \otimes \n - \I/3\right)$ with~$\I$ being the $3\times 3$ identity matrix,
$s$ is real and~$\n\in \mathbb{S}^2$, a unit vector. The vector, $\n$, is the eigenvector with the non-degenerate eigenvalue, known as the ``director'' and models the single preferred direction of uniaxial nematic alignment at every point in space~\cite{virga,dg}. The scalar, $s$, is known as the order parameter, which measures the degree of orientational order about $\n$.

In the absence of surface energies, a particularly simple form of the LdG energy is given by
\begin{equation}
    I_{LdG}[\Q]:=\int \frac{L}{2}|\nabla\Q|^2 + f_B\left(\Q\right) \mathrm{dA},\label{eq:3Denergy}
\end{equation}
where
\begin{equation}
    |\nabla\Q|^2:=\frac{\partial Q_{ij}}{\partial r_k}\frac{\partial Q_{ij}}{\partial r_k}, f_B\left(\Q\right):=\frac{A}{2}tr\Q^2-\frac{B}{3}tr\Q^3+\frac{C}{4}\left(tr\Q^2\right)^2.
    \label{eq:fB}
\end{equation}
The variable $A = \alpha\left(T-T^*\right)$ is a rescaled temperature, $\alpha, L, B, C >0$ are material-dependent constants, and $T^*$ is the characteristic nematic supercooling temperature. Further $\r:=\left(x,y, z\right)$, $tr\Q^2 = Q_{ij}Q_{ij}$ and $tr\Q^3 = Q_{ij}Q_{jk}Q_{ki}$ for $i,j,k = 1,2,3$.
The rescaled temperature $A$ has three characteristic values:(i)$A = 0$, below which the isotropic phase $\Q = 0$ loses stability, (ii) the nematic-isotropic transition temperature, $A = B^2/27C$, at which $f_B$ is minimized by the isotropic phase and a continuum of uniaxial states with $s = s_+ = B/3C$ and $\n$ arbitrary, and (iii) the nematic superheating temperature, $A = B^2/24C$ above which the isotropic state is the unique critical point of $f_B$.

 For a given $A<0$, let $\mathscr{N}:=\{\Q \in S_0:\Q=s_+\left(\n\otimes \n-\I/3\right)\}$ denote the set of minima of the bulk potential, $f_B$ with
\begin{equation}
    s_+:=\frac{B+\sqrt{B^2+24|A|C}}{4C}\nonumber
\end{equation}
and $\n\in S^2$ arbitrary. In particular, this set is relevant to our choice of Dirichlet conditions for boundary-value problems in what follows.
  The size of defect cores is typically inversely proportional to $s_+$ for low temperatures $A<0$.  Following \cite{wojtowicz1975introduction}, we use MBBA as a representative NLC material and use its reported values for $B$ and $C$ to fix $B=0.64\times10^4 N/m^2$ and $C=0.35\times10^4 N/m^2$ throughout this manuscript.

We use the one-constant approximation in (\ref{eq:fB}), so that the elastic energy density simply reduces to the Dirichlet energy density $|\nabla \Q|^2$. In general, the elastic energy density has different contributions from different deformation modes e.g. splay, twist and bend, and the elastic anisotropy can be strong for polymeric materials \cite{wensink2019polymeric}. However, the one-constant approximation assumes that all deformation modes have comparable energetic penalties i.e. equal elastic constants and this is a good approximation for some characteristic NLC materials such as MBBA \cite{dg},\cite{virga1995variational}, which makes the mathematical analysis more tractable.

We model nematic profiles on three-dimensional wells, whose cross section is a regular two-dimensional polygon $\Omega$, in the limit of vanishing depth, building on a batch of papers on square and rectangular domains \cite{canevari2017order, wang2019order,canevari_majumdar_wang_harris, kralj2014order}. More precisely, the domain is
\begin{equation}
\label{eq:domain}
\mathcal{B} = \Omega \times \left[0, h \right].
\end{equation}  $\Omega$ is a regular rescaled polygon, $E_K$, for example $E_6$ in Figure \ref{domain},  with $K$ edges, centered at the origin with vertices
\begin{equation}
    w_k =  \left(\cos\left(2\pi \left(k-1\right)/K\right),\sin\left(2\pi \left(k-1\right)/K\right)\right),\ k = 1,...,K.\nonumber
\end{equation}
We label the edges counterclockwise as $C_1, ..., C_K$, starting from $\left(1,0\right)$.
We work in the $h \to 0$ limit i.e. the thin film limit. Informally speaking, we impose Dirichlet uniaxial tangent boundary conditions on the lateral surfaces, which require the corresponding uniaxial director, $\mathbf{n}$, to be tangent to the lateral surfaces, and impose surface energies, $f_s$, on the top and bottom surfaces, which favour planar degenerate boundary conditions or equivalently constrain the nematic directors to be in the plane of the cross-section without a fixed direction. The Dirichlet conditions on the lateral sides are consistent with the tangent boundary conditions on the top and bottom surfaces.

In the $h \to 0$ limit and for certain choices of the surface energies, we can rigorously justify the reduction from the three-dimensional domain $\mathcal{B}$ to the two-dimensional domain $\Omega$ in (\ref{eq:domain}) \cite{Golovaty2015Dimension}. Firstly, we non-dimensionalize the system as, $\bar{\r} =  \left(\frac{x}{\lambda}, \frac{y}{\lambda}, \frac{z}{h} \right)$,
where $\lambda$ is the edge length of the regular polygon. We impose a Dirichlet boundary condition, $\Q_b$, on the lateral surfaces, $\partial \Omega \times \left[0, 1 \right]$ and assume that: 
\begin{equation} \label{BC}
 \Q\left(x, \, y, \, z\right) = \Q_b\left(x, \, y\right) \qquad \textrm{for }
 \left(x, \, y\right)\in\partial\Omega, \ z\in\left(0, \, 1\right) \qquad \textrm{and}
\end{equation}
\begin{equation} \label{hp-BC}
 \z \textrm{ is an eigenvector of } \Q_b\left(x, \, y\right)  \qquad\nonumber
 \textrm{ for any } \left(x, \, y\right)\in\partial\Omega\times \left(0, \, 1\right).\nonumber
\end{equation}
Then one can show (also see \cite{wang2019order}) that in the $\sigma = \frac{h}{\lambda}\to 0$  limit, minima of the Landau-de Gennes energy (\ref{eq:3Denergy})
 subject to the boundary condition~\eqref{BC} converge (weakly in~$H^1$)
 to minima of the reduced functional
\begin{equation}
\label{eq:reduced}
    F_0[\Q] := \int_{\Omega} \left(\frac{1}{2}\left|\nabla_{x,y}\Q\right|^2
    + \frac{\lambda^2}{L} f_B\left(\Q\right)\right) \mathrm{dA}
\end{equation}
 \emph{subject to the constraint} that
 \begin{equation} \label{constraint}
  \z \textrm{ \emph{is an eigenvector of} } \Q\left(x, \, y\right)  \qquad\nonumber
  \textrm{ for any } \left(x, \, y\right)\in\Omega\nonumber
 \end{equation}
 and to the boundary condition
 \[
  \Q = \Q_b \qquad \textrm{on } \partial\Omega.
 \]

Using the reasoning above, we restrict ourselves to $\Q$-tensors with $\z$ as a fixed eigenvector (this utilises two degrees of freedom for the allowed eigenvectors) and study critical points or minima of (\ref{eq:reduced}) with three degrees of freedom as shown below. 
\begin{equation}
    \begin{aligned}
        \Q\left(x,y\right) &= q_1\left(x,y\right)\left(\x\otimes\x-\y\otimes\y\right) + q_2\left(x,y\right)\left(\x\otimes\y+\y\otimes\x\right)\\
        &+ q_3\left(x,y\right)\left(2\z\otimes\z-\x\otimes\x-\y\otimes\y\right)
    \label{eq:Q}
    \end{aligned}
\end{equation}
where $\x= \left(1,0,0\right)$, $\y = \left(0,1,0\right)$ and $\z = \left(0,0,1\right)$.
Informally speaking, $q_1$ and $q_2$ measure the degree of ``in-plane" order, $q_3$  measures the ``out-of-plane" order and $\Q$ is invariant in the $z$-direction.
This constraint naturally excludes certain solutions such as the stable escaped (E) solution in a cylinder with large radius in \cite{shams2012theory}, for which the $z$-invariance does not hold.
In \cite{canevari_majumdar_wang_harris}, the authors compute bounds for $q_3$ as a function of the re-scaled temperature. In particular, they show that for $A = -\frac{B^2}{3C}$, $q_3$ is necessarily a constant so that critical points of the form (\ref{eq:Q}) only have two degrees of freedom, which makes the mathematical analysis more tractable. For arbitrary $A<0$, LdG critical points of the form (\ref{eq:Q}), subject to the Dirichlet boundary condition $\Q_b \in \mathscr{N}$, would have non-constant $q_3$ profiles and whilst we conjecture that some qualitative solution properties are universal for $A<0$, a non-constant $q_3$ profile would introduce new technical difficulties that would distract from the main message. A further benefit is that whilst we present our results in a 2D framework, these reduced critical points survive for all $h>0$ (beyond the thin-film limit) although they may not be physically relevant or energy-minimizing outside the thin-film limit (\cite{canevari2017order} and \cite{wang2019order}).

From \cite{canevari_majumdar_wang_harris}, for $A = -B^2/3C$, we necessarily have $q_3 = -\frac{B}{6C}$ and for all $\lambda>0$, the study of $\Q$ in (\ref{eq:Q}) is reduced to a symmetric, traceless $2\times 2$ matrix $\P$ given below -
\begin{equation}
    \P =
    \left(\begin{tabular}{cc}
    $P_{11}$&$P_{12}$\\
    $P_{12}$&$-P_{11}$\\
\end{tabular}\right).
\nonumber
\end{equation}
The relation between $\Q$ and $\P$ is
\begin{equation}
    \Q =
    \left(\begin{tabular}{cc|c}
        \multicolumn{2}{c|}{\multirow{2}*{$\P\left(\r\right)+\frac{B}{6C}\I_2$}} & $0$ \\
    \multicolumn{2}{c|}{} & $0$ \\
    \hline
    $0$ & $0$ & $-B/3C$ \\
    \end{tabular}\right).
    \label{eq:QP}
\end{equation}
    Therefore, the energy in (\ref{eq:reduced}) is reduced to
    \begin{equation}
        F[P]: = \int_{\Omega}\frac{1}{2}|\nabla P|^2+\frac{\lambda^2}{L}\left(-\frac{B^2}{4C}tr\P^2+\frac{C}{4}\left(tr\P^2\right)^2\right) \mathrm{d A},
        \label{p_energy}
    \end{equation}
and the corresponding Euler-Lagrange equations are
\begin{equation}
    \begin{aligned}
        \Delta P_{11} &= \frac{2C\lambda^2}{L}\left(P_{11}^2+P_{12}^2-\frac{B^2}{4C^2}\right)P_{11},\\
        \Delta P_{12} &= \frac{2C\lambda^2}{L}\left(P_{11}^2+P_{12}^2-\frac{B^2}{4C^2}\right)P_{12}.\\
\end{aligned}
    \label{Euler_Lagrange}
\end{equation}
We can also write $\P$ in terms of an order parameter $s$ and an angle $\gamma$ as shown below -
\begin{equation}
    \P = 2s\left(\n\otimes\n-\frac{1}{2}\I_2\right),
    \label{P}
\end{equation}
where $\n = \left(\cos\gamma,\sin\gamma\right)^T$ and $I_2$ is the $2\times 2$ identity matrix.
so that
\[P_{11} = s\cos\left(2\gamma\right),\ P_{12} = s\sin\left(2\gamma\right).\]
\begin{figure}
    \begin{center}
         \includegraphics[width=0.3\columnwidth]{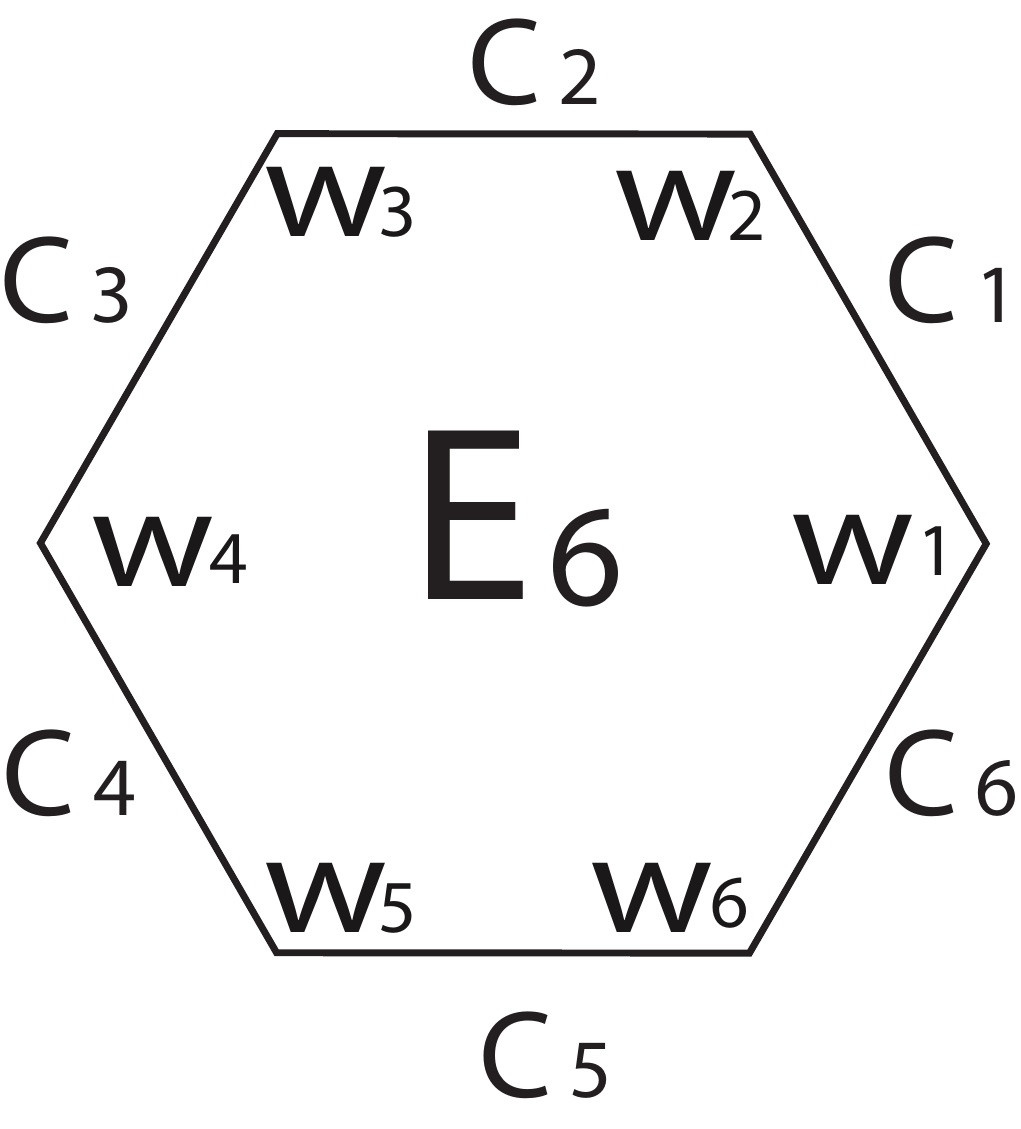}
        \caption{The regular rescaled hexagon domain $E_6$.}
        \label{domain}
    \end{center}
\end{figure}
   We briefly remark on the biaxiality parameter, $\beta(\Q) = 1-6\frac{tr(\Q^3)^2}{tr(\Q^2)^3}$ \cite{mkaddem2000fine}, where $\beta(\Q)\in[0,1]$ and $\beta(\Q) = 0$ for the uniaxial case. We can recover biaxiality in this reduced framework by using the relation between $\P$ and $\Q$ in (\ref{eq:QP}). When $\P=0$, the eigenvalues of $\Q$ are $(B/6C,B/6C,-B/3C)$ and $\beta(\Q)=0$ i.e. the nodal set of $\P$ defines a uniaxial set of $\Q$ with negative order parameter.

Next, we specify Dirichlet boundary conditions for $\P$ on $\partial E_K$. We work with tangent boundary condition on $\partial E_K$ which requires $\n$ in (\ref{P}) to be tangent to the edges of $E_K$, constraining the values of $\gamma$ on $\partial E_K$. However, there is a necessary mismatch at the corners/vertices.
We define the distance between a point on the boundary and the vertices as
\begin{equation}
    dist\left(w\right) = min\{||w-w_k||_2,k = 1,...,K\},\ w\ on\ \partial E_K.\nonumber
\end{equation}
We define the Dirichlet boundary condition $\P = \P_b$ on the segments of edges, far from the corners, as
\begin{equation}
    \begin{aligned}
    &P_{11b}\left(w\right) = \alpha_k = -\frac{B}{2C}\cos\left(\frac{\left(2k-1\right)2\pi}{K}\right),\ dist\left(w\right)>\epsilon, w\ on\ \partial E_K,\\
    &P_{12b}\left(w\right) = \beta_k = -\frac{B}{2C}\sin\left(\frac{\left(2k-1\right)2\pi}{K}\right),\ dist\left(w\right)>\epsilon, w\ on\ \partial E_K,
    \end{aligned}
    \label{Pb}
\end{equation}
   where $0<\epsilon \ll 1/2$ is the size of mismatch region. Recalling $\Q_b$ in (\ref{BC}), we have
\[
\Q_b = \P_b  - \frac{B}{6C}\left(2 \z\otimes \z - \x\otimes \x -
\y\otimes \y \right)
\]
which defines a Dirichlet uniaxial boundary condition, $\beta(Q_b)=0$, that is a minimizer of the bulk potential $f_B$ in (\ref{eq:fB}). At each vertex, we set $\P_b$ to be equal to the average of the two constant values on the two intersecting edges at the vertex under consideration. On the $\epsilon$-neighbourhood of the vertices, we linearly interpolate between the constant values in (\ref{Pb}) and the average value at the vertex and for $\epsilon$ sufficiently small, the choice of the interpolation does not change the qualitative solution profiles. In the next sections, we study minima of (\ref{p_energy}) as a function of $\lambda$, using a combination of analytic and numerical tools, with the hexagon as an illustrative example.

\section{Distinguished Limits}
\label{sec:distinguished_limit}
There is one parameter in the reduced energy (\ref{p_energy}) proportional to
\[
\bar{\lambda}^2 = \frac{2C \lambda^2}{L},
\]
 which is effectively the square of the ratio of two length scales, $\lambda$ and $\sqrt{\frac{L}{C}}$. Since we work at a fixed temperature, $A = -\frac{B^2}{3C}$ and we treat $B$, $C, L$ to be fixed material dependent constants, it is clear that $\frac{L}{C}$ is proportional to $\xi^2 = \frac{L}{|A|}$, where $\xi$ is a material-dependent and temperature-dependent characteristic length scale \cite{kralj2014order}. The length scale, $\xi$, is often referred to as the nematic correlation length and is typically associated with defect core sizes. The nematic correlation length is typically in the range of a few tens to hundreds of nanometers \cite{virga}.

We study two distinguished limits analytically in what follows - the $\bar{\lambda} \to 0$ limit is relevant for nano-scale domains $\Omega$, and the $\bar{\lambda}\to \infty$ limit, which is the macroscopic limit relevant for micron-scale or larger cross-sections $\Omega$. We present rigorous results for limiting problems below but our numerical simulations show that the limiting results are valid for non-zero but sufficiently small $\bar{\lambda}$ (or even experimentally accessible nano-scale geometries depending on parameter values) and sufficiently large but finite $\bar{\lambda}$ too. In other words, these limiting results are of potential practical value too.
We treat $C$ and $L$ as fixed constants in this manuscript and hence, the $\bar{\lambda}\to 0$ and $\bar{\lambda}\to\infty$  limits are equivalent to the $\lambda \to 0$ and $\lambda\to\infty$ limits respectively. In the following, we drop the bar over $\lambda$ for brevity.

\subsection{The $\lambda\to 0$ Limit}
\label{subsec:0_limit}
We can use Lemma 8.2 of \cite{lamy2014bifurcation} to deduce that there exists a $\lambda_0\left(B, C, L\right) > 0$ such that, for any $\lambda <\lambda_0\left(B, C, L\right)$, the system (\ref{Euler_Lagrange}) has a unique solution which is the unique minimizer
of the reduced energy in (\ref{p_energy}).

In \cite{kralj2014order} and \cite{canevari2017order}, the authors report the Well Order Reconstruction Solution (WORS) on a square domain, for all $\lambda >0$. The WORS is represented by a $\Q$-tensor of the form
\[
\Q_{WORS} = q \left(\x\otimes \x - \y\otimes \y \right) - \frac{B}{6C}\left(2\z\otimes \z - \x\otimes \x - \y\otimes \y \right)
\]
where $q$ is a scalar function such that $q=0$ along the square diagonals. Mathematically speaking, this implies that the $\Q_{WORS}$ is strictly uniaxial with negative order parameter along the square diagonals which would manifest as a pair of orthogonal defect lines in experiments. The WORS is globally stable for small $\lambda$ and loses stability as $\lambda$ increases. Numerical experiments suggest that the WORS acts as a transition state between experimentally observable equilibria for large $\lambda$.

It is natural to study the counterparts of the WORS on arbitrary regular two-dimensional polygons, $E_K$, and in particular study the zero set of the corresponding $\P$ matrix in (\ref{eq:QP}). Namely, is the zero set of $\P$ a set of intersecting lines as in the WORS or it is a lower-dimensional set of discrete or unique points? We address this question below by means of an explicit analysis of the limiting problem with $\lambda = 0$.

We define the limiting problem for $\lambda = 0$ to be
\begin{equation}
\begin{aligned}
    &\Delta P_{11}^0 = 0,\ \Delta P_{12}^0 = 0, on\ \Omega,\\
    &P_{11}^0 = P_{11b},\ P_{12}^0 = P_{12b},\ on\ \partial\Omega.
\end{aligned}
    \label{zero_euler}
\end{equation}
We can adapt methods from \cite{bethuel1994ginzburg} and from Proposition $3.1$ of \cite{fang2019solution}, we have that minima, $\left(P_{11}^{\lambda},P_{12}^{\lambda}\right)$, of (\ref{p_energy}) subject to the fixed boundary conditions $\P_b$ in (\ref{Pb}) (for $\epsilon$ sufficiently small) converge uniformly to the unique solution $\left(P_{11}^0,P_{12}^0\right)$ of (\ref{zero_euler}) as $\lambda \to 0$ i.e.
\begin{equation}
    |P_{11}^{\lambda}-P_{11}^0|_{\infty}\leq C\lambda^2,\ |P_{12}^{\lambda}-P_{12}^0|_{\infty}\leq C\lambda^2,
    \label{eq:error_estimates}
\end{equation}
for $C$ independent of $\lambda$.
Therefore, in the $\lambda\to 0$ limit, it suffices to study the boundary-value problem for the Laplace equation in (\ref{zero_euler}) on regular polygons.
\subsubsection{Solving Laplace equation with Dirichlet boundary conditions on a regular polygon domain}
\begin{figure}
    \begin{center}
        \includegraphics[width=0.8\columnwidth]{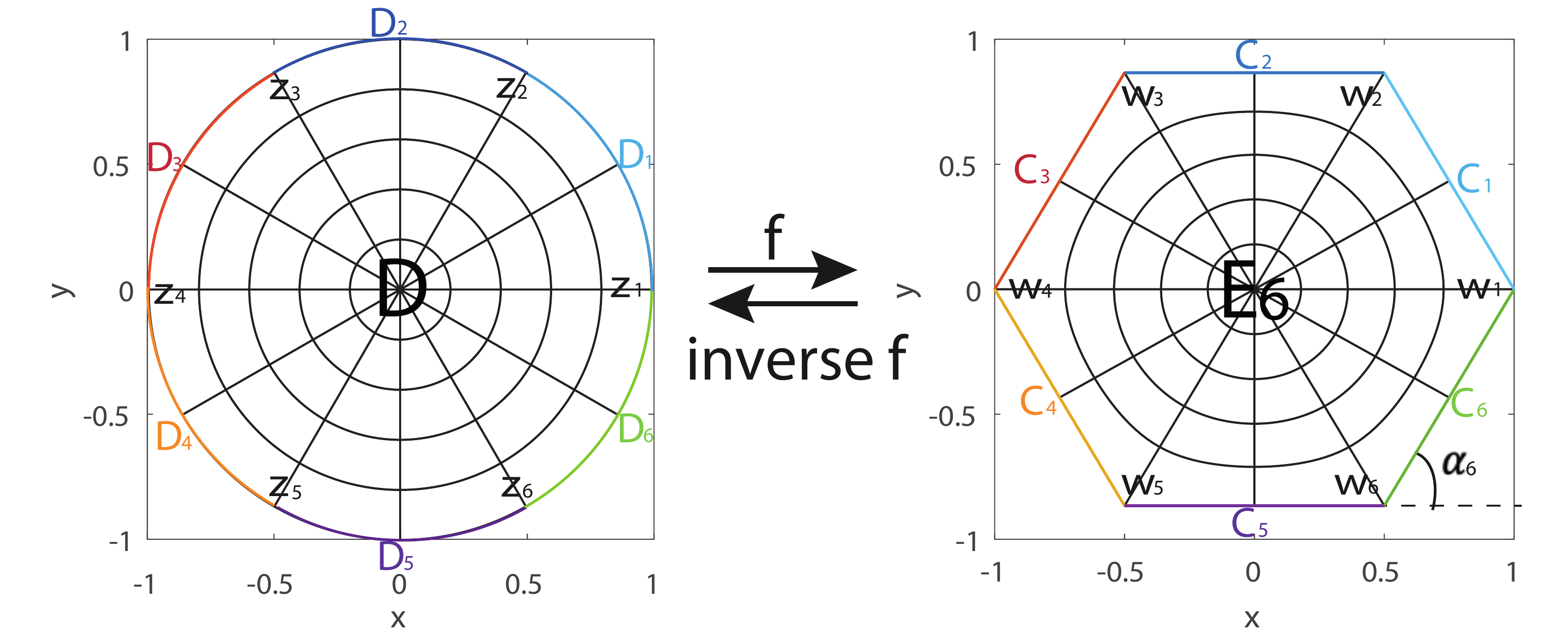}
        \caption{Schwarz-Christoffel mapping $f$ from a unit disc to a regular hexagon and inverse mapping $f^{-1}$ from a regular hexagon to a unit disc.}
        \label{circle_6_hexagon}
    \end{center}
\end{figure}
Our strategy is to map the Dirichlet boundary-value problem (\ref{zero_euler}) on $\Omega = E_K$ (a regular polygon with $K$ edges) to an associated Dirichlet boundary-value problem on the unit disc $D$ in Figure \ref{circle_6_hexagon}, for which the solution can be easily computed by the Poisson Integral \cite{krantz2012handbook}.
In complex analysis, a Schwarz-Christoffel mapping is a conformal transformation, $f: D \to E_K$ of the disc (upper half-plane) onto the interior of any simple polygon (the boundary of the polygon does not cross itself) \cite{driscoll2002schwarz}, such that $f\left(D\right) = E_K$. Let $w = f\left(z\right)$. We require that $f\left(z_k\right) = w_k = e^{i 2\pi (k-1)/K}$, $f\left(0,0\right) = \left(0,0\right)$ and $f^{-1}\left(w_1\right) = z_1 = \left(1, 0 \right)$. Then $z_k =e^{i 2\pi (k-1) /K}$ and exterior angles of the $E_K$ along $C_{k-1}$ and $C_{k}$ are $\alpha_k=\frac{2\pi}{K}$, for $k = 1,...,K$.
The Schwarz-Christoffel mapping is uniquely determined as
\cite{brilleslyper2012explorations}
\begin{equation}
    f\left(z\right) = C_1\left(K\right)\int_0^z\frac{1}{\left(1-x^K\right)^{2/K}}\rm{d}x\nonumber
\end{equation} with
\begin{equation}
    C_1\left(K\right) = \frac{\Gamma\left(1-1/K\right)}{\Gamma\left(1+1/K\right)\Gamma\left(1-2/K\right)}.\nonumber
\end{equation}
The Taylor series representation of $f\left(z\right)$ is
\begin{equation}
    w = f\left(z\right) = C_1\left(K\right)\sum_{n = 0}^{\infty}{n-1+2/K\choose n}\frac{z^{1+nK}}{1+nK}.\nonumber
\end{equation}
The inverse of a conformal mapping, $f$, is also a conformal mapping, $f^{-1}$.
The conformal mapping, $f$, from a unit disc onto a regular hexagon and the inverse mapping, $f^{-1}$, from a regular hexagon to a unit disc, as example, is shown in Figure \ref{circle_6_hexagon}. One can check that $f$ maps the circle, $\partial D$, onto the polygon boundary, $\partial E_K = f\left(\partial D\right)$.\\
We define the disc boundary segments as
\begin{equation}
    D_k: = \{z = e^{i\theta}, 2\pi \left(k-1\right)/K \leq \theta < 2\pi k/K\},\ k = 1,...,K.\nonumber
\end{equation}
Then we can check that
\begin{equation}
    f\left(D_k\right) = C_k,\ f\left(\rho e^{\pi ki/K}\right) = \lambda e^{\pi ki/K},\ k = 1,...,K,\nonumber
\end{equation}
where $C_k$ is the $k$-th edge of $E_K$ and the last relation comes from
\begin{align}
    f\left(\rho e^{\pi ki/K}\right) &= C_1\left(K\right)\sum_{n = 0}^{\infty}{n-1+2/K\choose n}\frac{e^{\pi ki/K}e^{nki \pi}}{1+nK}\nonumber\\
    &= e^{\pi ki/K}C_1\left(K\right)\sum_{n = 0}^{\infty}{n-1+2/K\choose n}\frac{\left(-1\right)^{nk}}{1+nK}\nonumber\\
    &= \lambda e^{\pi ki/K},\nonumber
\end{align}
since $C_1\left(K\right)$ is real. $f$ is well defined on $\overline{D}$ and analytic in $\overline{D}/\ \{z_1,...,z_K\}$, whereas it is not smooth at $z_1,...z_K$ because there is a jump of $arg\frac{1}{\left(x-z_k\right)^{\alpha_k/\pi}}$\cite{driscoll2002schwarz}.
$f$ can be extended continuously to $\overline{D}$ at each $z_k$.

In complex analysis, let $u:U\rightarrow\mathbb{R}$ be a harmonic function in a neighborhood of the closed disc $\overline{D}\left(0,1\right)$, then for any point $z_0 = \rho e^{i\phi}$ in the open disc $D\left(0,1\right)$,
\begin{equation}
    u\left(\rho e^{i\phi}\right) = \frac{1}{2\pi}\int^{2\pi}_{0} u\left(e^{i\theta}\right)\frac{1-\rho^2}{1-2\rho \cos\left(\theta-\phi\right)+\rho^2}\mathrm{d\theta}.\nonumber
\end{equation}
If the Dirichlet boundary condition is piecewise constant (as in our case with $\epsilon =0$) on the segments $D_k$,
\begin{equation}
    u\left(\rho e^{i\phi}\right)= \frac{1}{2\pi}\sum_{k = 1}^{K} \int_{D_k} d_k\frac{1-\rho^2}{1-2\rho \cos\left(\theta-\phi\right)+\rho^2} \mathrm{d\theta} = \frac{1}{\pi}\sum_{k=1}^{K} d_kS_k\left(\rho e^{i\phi}\right),
    \label{Poisson_Integral}
\end{equation}
where $d_k$ is the constant boundary value on $C_k$ and $D_k$. To calculate $S_k$, we need to compute the integral
\begin{equation}
    I = \int\frac{1}{1+\rho^2-2\rho \cos x} \mathrm{dx}.
\end{equation}
Using a change of variable $t = \tan\frac{x}{2}$, we find that
\begin{equation}
    I = \int\frac{1}{1+\rho^2-2\rho\left((1-t^2)/(1+t^2)\right)}\frac{2\rm{d}t}{1+t^2}
         = \frac{2}{1-\rho^2}\left(\arctan \left(\frac{1+\rho}{1-\rho}\tan\frac{x}{2}\right)+const \right)\nonumber
\end{equation}
If the angle $2\pi \left(k-1\right)/K-\phi\leq \left(2n+1\right)\pi < 2\pi k/K-\phi$, $n \in \mathbb{Z}$, $k = 1,...,K$, $S_k = \int_{2\pi\left(k-1\right)/K}^{2\pi k/K}$ is an improper integral
\cite{ma}
and
\begin{equation}
    \begin{aligned}
        S_{k}\left(\rho e^{i\phi}\right) &= \frac{1-\rho^2}{2}\left(I|_{x = 2\pi k/K-\phi}-I|_{x\to \left(2n+1\right)\pi^+} \right.\\
        & \left. + I|_{x\to \left(2n+1\right)\pi^-}-I|_{x = 2\pi \left(k-1\right)/K-\phi}\right)\\
        & = \arctan\left(\frac{1+\rho}{1-\rho}\tan\frac{2\pi k/K-\phi}{2}\right)\\
        & - \arctan\left(\frac{1+\rho}{1-\rho}\tan\frac{2\pi \left(k-1\right)/K-\phi}{2}\right)+\pi\\
    \end{aligned}
    \label{Sk_1}
\end{equation}
otherwise,
\begin{equation}
    \begin{aligned}
        &S_{k}\left(\rho e^{i\phi}\right) = \frac{1-\rho^2}{2}\left(I|_{x = 2\pi k/K-\phi}-I|_{x = 2\pi \left(k-1\right)/K+\phi}\right)\\
         &= \arctan\left(\frac{1+\rho}{1-\rho}\tan\frac{2\pi k/K-\phi}{2}\right) - \arctan\left(\frac{1+\rho}{1-\rho}\tan\frac{2\pi \left(k-1\right)/K-\phi}{2}\right)
    \end{aligned}
    \label{Sk_2}
\end{equation}
Equation (\ref{Poisson_Integral}) is Poisson Integral on unit disc and $u\left(z\right)$ is a harmonic function of $z$ on the unit disc $D$. If we consider the conformal mapping, $z = f^{-1}\left(w\right)$, then $U\left(w\right) = u\left(f^{-1}\left(w\right)\right)$ is a harmonic function of $w$ on $E_K$, subject to specified Dirichlet conditions on the edges $C_K$ of $E_K$.
The proof can be found in Proposition 6.1 of \cite{olver2017complex}.

\subsubsection{Ring Solutions for $\lambda=0$}
We can use the Poisson formula in Equation~(\ref{Poisson_Integral}) to explicitly compute the solution of the boundary-value problem (\ref{zero_euler}). In the $\epsilon \to 0$ limit, the solution of (\ref{zero_euler}) converges uniformly to the solution of
the boundary-value problem below, with piecewise constant boundary conditions
\begin{equation}
\begin{aligned}
    &\Delta P_{11}\left(\r\right)  = 0,\ \r\in E_K,\\
        &\Delta P_{12}\left(\r\right)  = 0,\ \r\in E_K,\\
        &P_{11}\left(\r\right)  = \alpha_k =-\frac{B}{2C}\cos\left(\left(2k-1\right)2\pi/K\right),\ \r\ on\ C_k,\ k = 1,..., K.\\
        &P_{12}\left(\r\right)  = \beta_k = -\frac{B}{2C}\sin\left(\left(2k-1\right)2\pi/K\right),\ \r\ on\ C_k,\ k = 1,...,K.
\end{aligned}
\label{P11_P12_laplace}
\end{equation}
For simplicity, we focus on the boundary-value problem, (\ref{P11_P12_laplace}) with piecewise constant boundary conditions.
\begin{proposition} Let $\left(P_{11}, P_{12} \right)$ be the unique solution of (\ref{P11_P12_laplace}) and let
    \begin{equation}
        G_K:=\{S\in O\left(2\right):SE_K\in E_K\},
        \label{sym-set}
    \end{equation}
be a set of symmetries consisting of
$K$ rotations by angles $2\pi k/K$ for $k = 1,...,K$ and $K$ reflections about the symmetry axes ($\phi=\pi k/K$, $k = 1,...,K$) of the polygon $E_K$.\\
$P_{11}^2+P_{12}^2$ is invariant under $G_K$.
    If $\left(P_{11},P_{12}\right)\neq (0,0)$, then
    $\frac{\left(P_{11},P_{12}\right)}{\sqrt{P_{11}^2+P_{12}^2}}$ undergoes a reflection about the symmetry axes of the polygon and rotates by $4\pi k/K$ under rotations of angle $2\pi k/K$ for $k = 1,...,K$.
    \label{symmetry_proposition}
\end{proposition}
\begin{proof}
For convenience, we extend the definition of $S_k$, $\alpha_k$, $\beta_k$, $k = 1,...,K$, to $k\in \mathbb{Z}$ and use the periodicity of $\tan$, $\cos$ and $\sin$ to define
\begin{equation}
        S_{k+nK} = S_k,\ \alpha_{k+nK} = \alpha_k,\ \beta_{k+nK} = \beta_k,\ n\in\mathbb{Z}. \label{eq:periodic}
\end{equation}
From the definitions in (\ref{Sk_1}) and (\ref{Sk_2}),
    \begin{align}
        &S_j\left(\rho e^{i\phi + 2\pi ki/K}\right) = S_{j-k}\left(\rho e^{i\phi}\right),\label{eq:s1}\\
        &S_j\left(\rho e^{-i\phi}\right) = S_{1-j}\left(\rho e^{i\phi}\right),\ j\in\mathbb{Z},\ k\in\mathbb{Z} \label{eq:s2},
    \end{align}
and from the definition of $\alpha_k$ and $\beta_k$ in \ref{P11_P12_laplace}, we have
\begin{equation}
    \begin{aligned}
        \alpha_{j+k} &= \alpha_j\cos\left(\frac{4\pi k}{K}\right)-\beta_j\sin\left(\frac{4\pi k}{K}\right),\\
        \beta_{j+k} &= \beta_j\cos\left(\frac{4\pi k}{K}\right)+\alpha_j\sin\left(\frac{4\pi k}{K}\right),\ j\in\mathbb{Z},\ k\in\mathbb{Z}
    \end{aligned}
    \label{alpha_beta}
\end{equation}
       and
\begin{equation} \label{eq:ab}
    \alpha_j  = \alpha_{1-j}; \quad \beta_j =-\beta_{1-j},\ j\in\mathbb{Z}.
\end{equation}
Let $\left(p_{11}, p_{12}\right)$ be the solution of the Laplace equation on the unit disc, subject to the boundary conditions, $p_{11} = \alpha_k$ and $p_{12} = \beta_k$ on the disc segment $D_k$. From (\ref{Poisson_Integral}), (\ref{eq:s1}) and (\ref{alpha_beta}), we have
\begin{eqnarray}
    p_{11}\left(\rho e^{i\phi+2\pi ki/K}\right) &&= \frac{1}{\pi}\sum_{j = 1}^{K}\alpha_j
    S_j\left(\rho e^{i\phi+2\pi ki/K} \right) = \frac{1}{\pi}\sum_{j = 1-k}^{K-k}\alpha_{j+k}S_j\left(\rho e^{i\phi} \right) \nonumber\\
      && = \frac{1}{\pi}\sum_{j = 1}^K\alpha_j S_j\left(\rho e^{i\phi}\right)\cos\left(\frac{4\pi k}{K}\right)-\frac{1}{\pi}\sum_{j = 1}^K\beta_j S_j\left(\rho e^{i\phi}\right)\sin\left(\frac{4\pi k}{K}\right) \nonumber \\
     &&= p_{11}\left(\rho e^{i\phi}\right)\cos\left(\frac{4\pi k}{K}\right) - p_{12}\left(\rho e^{i\phi}\right)\sin\left(\frac{4\pi k}{K}\right). \label{eq:p11}
\end{eqnarray} Here, we use (\ref{eq:periodic}) to manipulate the limits of the summation above. Similarly,
\begin{eqnarray}
    p_{12}\left(\rho e^{i\phi+2\pi ki/K}\right)
    = p_{12}\left(\rho e^{i\phi}\right)\cos\left(\frac{4\pi k}{K}\right) + p_{11}\left(\rho e^{i\phi}\right)\sin\left(\frac{4\pi k}{K}\right). \label{eq:p12}
\end{eqnarray}
We can use (\ref{eq:p11}) and (\ref{eq:p12}) to check that $p_{11}^2+p_{12}^2=s^2$ is invariant under rotations by multiples of $2\pi k/K$ and
$\frac{\left(p_{11},p_{12}\right)}{\sqrt{p_{11}^2+p_{12}^2}} $ rotates by $4\pi k/K$ under rotations by $2\pi k/K$, $k = 1,...,K.$
Similarly, we can use (\ref{Poisson_Integral}), (\ref{eq:ab}) and (\ref{eq:s2}) to show that
\begin{eqnarray}
    p_{11}\left(\rho e^{-i\phi}\right) &&=\frac{1}{\pi}\sum_{j = 1}^{K}\alpha_jS_j\left(\rho e^{-i\phi}\right)= \frac{1}{\pi}\sum_{j = 1}^{K}\alpha_j S_{1-j}\left(\rho e^{i\phi}\right) \nonumber\\
    &&= \frac{1}{\pi}\sum_{j = 1}^K\alpha_jS_j\left(\rho e^{i\phi}\right) = p_{11}\left(\rho e^{i\phi}\right)\label{eq:p11_0}
\end{eqnarray} and using analogous arguments,
\begin{eqnarray}
    p_{12}\left(\rho e^{-i\phi}\right)
      =-p_{12}\left(\rho e^{i\phi}\right).\label{eq:p12_0}
\end{eqnarray}
    We can use (\ref{eq:p11}), (\ref{eq:p12}), (\ref{eq:p11_0}) and (\ref{eq:p12_0}) to obtain the relation
\begin{eqnarray}
    & p_{11}\left(\rho e^{k\pi i/K - \phi i}\right) = p_{11}\left(\rho e^{-k\pi i/K + \phi i}\right) = p_{11}\left(\rho e^{k\pi i/K + \phi i-2k\pi i/K}\right)\nonumber\\
    & = p_{11}\left(\rho e^{k\pi i/K + \phi i}\right)\cos\left(\frac{-4k\pi}{K}\right) - p_{12}\left(\rho e^{k\pi i/K + \phi i}\right)\sin\left(\frac{-4k\pi}{K}\right)\nonumber\\
    & = p_{11}\left(\rho e^{k\pi i/K + \phi i}\right)\cos\left(\frac{4k\pi}{K}\right) + p_{12}\left(\rho e^{k\pi i/K + \phi i}\right)\sin\left(\frac{4k\pi}{K}\right).\nonumber \label{eq:p11_k}
\end{eqnarray} and using analogous arguments,
\begin{eqnarray}
    & p_{12}\left(\rho e^{k\pi i/K - \phi i}\right) =
    & = -p_{12}\left(\rho e^{k\pi i/K + \phi i}\right)\cos\left(\frac{4k\pi}{K}\right) + p_{11}\left(\rho e^{k\pi i/K + \phi i}\right)\sin\left(\frac{4k\pi}{K}\right).\nonumber \label{eq:p12_k}
\end{eqnarray}
Thus, $p_{11}^2+p_{12}^2=s^2$ is invariant under reflection about $\phi = k\pi i/K, k = 1,...,K$ and
$\frac{\left(p_{11},p_{12}\right)}{\sqrt{p_{11}^2+p_{12}^2}}$ is reflected across $\phi = k\pi i/K, k = 1,...,K$.
Since $f$ is a conformal mapping, it preserves rotation symmetry and reflection symmetry,
\begin{equation}
    \begin{aligned}
        f\left(\rho e^{i\phi}e^{2\pi ik/K}\right)&= f\left(\rho e^{i\phi}\right)e^{2\pi ik/K}, \nonumber\\
        f\left(\rho e^{-i\phi}\right) &= \overline{f\left(\rho e^{i\phi}\right)},\nonumber\\
    \end{aligned}
\end{equation}
We have $P_{11}\left(w\right) = p_{11} \left( f^{-1}\left(w\right) \right)$ and $P_{12}\left(w\right) = p_{12}\left( f^{-1}\left(w\right)\right)$ for $w \in E_K$, $P_{11}^2+P_{12}^2$ is invariant under the symmetries in the set $G_K$ and the vector,
$\frac{\left(P_{11},P_{12}\right)}{\sqrt{P_{11}^2+P_{12}^2}}$, is reflected about the symmetry axes of the polygon and rotates by $4\pi k/K$ under rotations of $2\pi k/K$ for $k = 1,...,K$.
\end{proof}

\begin{proposition} Let $\P_R= \left(P_{11}, P_{12}\right)$ be the unique solution of the boundary-value problem (\ref{P11_P12_laplace}). Then $P_{11}\left(0,0\right) = 0, P_{12}\left(0,0\right) = 0$ at the centre of all regular polygons, $E_K$.
However, $\P_R \left(x, y\right) \neq \left(0,0\right)$ for $\left(x,y\right)\neq \left(0,0\right)$, for all $E_K$ with $K \neq 4$ i.e. the WORS is a special case of $\P_R$ on $E_4$ such that $\P_R = \left(0,0\right)$ on the square diagonals. For $K\neq 4$, the origin is the unique zero of the unique solution $\P_R$, referred to as the ``ring solution" in the rest of the paper.
    \label{isotropic_center}
\end{proposition}
\begin{proof}
We set $\rho=0$ in (\ref{Poisson_Integral}) to compute $\left(P_{11}, P_{12} \right)\left(0,0\right) = \left(p_{11}, p_{12} \right)\left(f^{-1} \left(0,0\right) \right)$ as shown below, recalling that $f\left(0,0\right) = \left(0,0\right)$ i.e.
\begin{eqnarray}
    p_{11}\left(0,0\right) &&= \frac{1}{2\pi}\sum_{k = 1}^K\alpha_k\int_{D_k} \mathrm{d\theta} =
     \frac{1}{K}\sum_{k = 1}^K \alpha_k \nonumber \\
    && = -\frac{B}{2KC}\sum_{k = 1}^K \cos\left(\left(2k-1\right)2\pi/K\right) \nonumber\\
    && = -\frac{B}{2KC}\sum_{k = 1}^K \frac{\sin\left(\left(2k-1\right)2\pi/K+2\pi/K\right)-\sin\left(\left(2k-1\right)2\pi/K-2\pi/K\right)}{2\sin\left(2\pi/K\right)} \nonumber\\
    && = -\frac{B}{4KC \sin\left(2\pi/K\right)}\sum_{k = 1}^K \sin\left(4\pi k/K\right) - \sin\left(4\pi \left(k-1\right)/K\right) = 0\nonumber
    \label{eq:p10}
\end{eqnarray}
       and similarly, $p_{12}\left(0,0\right) = 0$.
Hence, we have $P_{11}\left(0,0\right) = P_{12}\left(0,0\right) = 0$ for any regular polygon, since $\left(0,0\right)$ is a fixed point of the mapping $f$.

Set $x = \frac{1 + \rho}{1 - \rho}$. For a fixed $\phi = \phi^*$, if $\frac{\partial p_{11}}{\partial x}\equiv 0$ for any $x \geq 1$, $p_{11}\equiv 0$ on $\phi =\phi^*$.
    Otherwise, if $\frac{\partial p_{11}}{\partial x} > 0\left(<0\right)$ for any $x > 1$, $p_{11}=0$ only at the center. Recalling (\ref{Poisson_Integral}), we have
\begin{eqnarray}
    &&p_{11}(\rho e^{i\phi}) = \sum_{k = 1}^{K}\frac{1}{\pi}\alpha_k S_k(\rho e^{i\phi}) \nonumber \\
    &&= \frac{B}{2\pi C}\sum_{k = 1}^{K}\arctan\left(x\tan\left(\pi k/K-\phi/2\right)\right)\left(\cos\frac{2\pi\left(2k+1\right)}{K}-\cos\frac{2\pi\left(2k-1\right)}{K}\right)+\alpha_{k_*} \nonumber\\
    &&= -\frac{B}{\pi C}\sin\frac{2\pi}{K}\sum_{k = 1}^{K}\arctan\left(x\tan\left(\pi k/K-\phi/2\right)\right)\left(\sin\frac{4\pi k}{K}\right)+\alpha_{k_*}\nonumber
\label{eq:p1c}
\end{eqnarray}
    where $\alpha_{k_*}$ is the boundary value on the segment for which $S_k$ is an improper integral (\ref{Sk_1}) i.e. $2\pi\left(k_*-1\right)/K\leq \phi + (2n+1)\pi<2\pi k_*/K$,$n\in\mathbb{Z}$. From Proposition \ref{symmetry_proposition}, it suffices to focus on the sector $0 \leq \phi \leq \frac{\pi}{K}$. Next, we define
    \begin{equation}
        K_{half} = \begin{cases}
            \frac{K-1}{2},&\ K\ is\ odd,\\
            \frac{K}{2}-1,&\ K\ is\ even,\\
        \end{cases}
        \nonumber
    \end{equation} and compute
    \begin{eqnarray}
        &&\frac{\partial p_{11}}{\partial x} = -\frac{B}{\pi C}\sin\frac{2\pi}{K}\sum_{k = 1}^{K}\frac{\tan\left(\pi k/K-\phi/2\right)}{1+\tan^2\left(\pi k/K-\phi/2\right)x^2}\sin\left(4\pi k/K\right) \nonumber\\
           &&=-\frac{B}{2\pi C}\sin\frac{2\pi}{K}\sum_{k=1}^{K_{half}}\left(\frac{\sin\left(2\pi k/K-\phi\right)}{1+\left(x^2-1\right)\sin^2\left(\pi k/K-\phi/2\right)}\right.\nonumber\\
           &&+\left.\frac{\sin\left(2\pi k/K+\phi\right)}{1+\left(x^2-1\right)\sin^2\left(\pi k/K+\phi/2\right)}\right)\sin\left(4\pi k/K\right).\nonumber
\label{derivation}
\end{eqnarray}
    When $x=1$, i.e., $\rho = 0$,  we obtain
    \begin{eqnarray}
        \left.\frac{\partial p_{11}}{\partial x}\right\vert_{x=1} &&= -\frac{B}{2\pi C}\sin\frac{2\pi}{K}\sum_{k=1}^{K_{half}}\left(\sin\left(2\pi k/K-\phi\right)+\sin\left(2\pi k/K+\phi\right)\right)\sin\left(4\pi k/K\right) \nonumber \\
        && = \frac{B}{2\pi C}\sin\frac{2\pi}{K}\cos\left(\phi\right)\sum_{k=1}^{K_{half}}\left(\cos\left(6\pi k/K\right)-\cos\left(2\pi k/K\right)\right).
   \label{eq:px}
    \end{eqnarray}

It is relatively straightforward to check using (\ref{eq:px}) that for $x = 1$,
\begin{equation}
    \frac{\partial p_{11}}{\partial x} = \begin{cases}
        0, & K\in\mathbb{Z},\ K>3;\nonumber\\
        \frac{3\sqrt{3}B}{8\pi C}\cos\phi,&\ K = 3.
    \end{cases}
    \label{equ_1}
\end{equation}
We can use (\ref{eq:px}) to study the sign of
$\left.\frac{\partial p_{11}}{\partial x}\right\vert_{x>1}$ as shown below.
When $x > 1$, $K = 3$, $0\leq \phi\leq \pi/3$, we have
\begin{eqnarray}
    &&\frac{\partial p_{11}}{\partial x}= \nonumber \\
     && -\frac{B}{2\pi C}\sin\frac{2\pi}{3}\left(\frac{\sin\left(2\pi/3-\phi\right)}{1+\left(x^2-1\right)\sin^2\left(\pi/3-\phi/2\right)}+\frac{\sin\left(2\pi/3+\phi\right)}{1+\left(x^2-1\right)\sin^2\left(\pi/3+\phi/2\right)}\right)\sin\left(4\pi/3\right) \nonumber \\
    & & > -\frac{B}{2\pi C}\sin\frac{2\pi}{3}\left(\sin\left(2\pi/3-\phi\right)+\sin\left(2\pi/3+\phi\right)\right)\sin\left(4\pi/3\right)/x^2 \nonumber \\
    &&=\left.\frac{\partial p_{11}}{\partial x}\right\vert_{x = 1}/x^2 = \frac{3\sqrt{3}B}{8\pi C}\frac{\cos\phi}{x^2} >0.\nonumber
    \end{eqnarray}
For $K = 4$, for any $x>1$, $0\leq\phi\leq \pi/4$,
\begin{eqnarray}
    &&\frac{\partial p_{11}}{\partial x}=-\frac{B}{2\pi C}\sin\frac{\pi}{2}\left(\frac{\sin\left(\pi/2-\phi\right)}{1+\left(x^2-1\right)\sin^2\left(\pi/4-\phi/2\right)}+\frac{\sin\left(\pi/2+\phi\right)}{1+\left(x^2-1\right)\sin^2\left(\pi/4+\phi/2\right)}\right)\sin\left(\pi\right) \nonumber\\
    && = 0\nonumber
\end{eqnarray}
Otherwise, for $K\in\mathbb{Z}$, $K>4$, $x>1$, we have
\begin{eqnarray}
    &&\frac{\partial p_{11}}{\partial x}
    <-\frac{B}{2\pi C}\sin\frac{2\pi}{K}\sum_{k=1}^{K_{half}}\left(\frac{\sin\left(2\pi k/K-\phi\right)}{1+\left(x^2-1\right)\sin^2\left(\theta_*\right)}+\frac{\sin\left(2\pi k/K+\phi\right)}{1+\left(x^2-1\right)\sin^2\left(\theta_*\right)}\right)\sin\left(4\pi k/K\right) \nonumber \\
    &&=\left.\frac{\partial p_{11}}{\partial x}\right\vert_{x =1}/\left(\cos^2\left(\theta_*\right)+x^2\sin^2\left(\theta_*\right)\right) = 0,\nonumber
\end{eqnarray}
where
\begin{equation}
    \theta_* = \begin{cases}
        \frac{\pi}{4}-\frac{\pi}{2K}, &\ K\ mod\ 4 = 0;\nonumber\\
        \frac{\pi}{4}+\frac{\pi}{4K}, &\ K\ mod\ 4 = 1;\\
        \frac{\pi}{4}, &\ K\ mod\ 4 = 2;\\
        \frac{\pi}{4}-\frac{\pi}{4K}, &\ K\ mod\ 4 = 3.\\
    \end{cases}
\end{equation}
Therefore when $x>1$, $0\leq\phi\leq \pi/K$
\begin{equation}
    \frac{\partial p_{11}}{\partial x} \begin{cases}
        >\ 0,&\ K = 3;\nonumber\\
        =\ 0,&\ K = 4;\\
        <\ 0, & K\in\mathbb{Z},\ K > 4;\\
    \end{cases}
    \label{>_1}
\end{equation} and by the symmetry results in Proposition~\ref{symmetry_proposition}, we have that $\frac{\partial p_{11}}{\partial x}$ is non-zero for $x>1$, $K \neq 4$ for any regular polygon $E_K$.
So $p_{11} = 0$ everywhere for the square domain and for $K \neq 4$, $p_{11}$ only vanishes at the origin. For any $K\geq 3$, when $\phi = 0$,
\begin{equation}
\begin{aligned}
    p_{12}(\rho) &= \sum_{k = 1}^{K}\frac{1}{\pi} \beta_k S_k(\rho)  = \frac{B}{\pi C}\sin\frac{2\pi}{K}\sum_{k = 1}^{K_{half}}\left\{\arctan\left(\frac{1+\rho}{1-\rho}\tan\left(\pi k/K\right)\right)\cos \frac{4\pi k}{K}\right. \\
    & + \left.\arctan\left(\frac{1+\rho}{1-\rho}\tan\left(\pi \left(K-k\right)/K\right)\right)\cos \frac{4\pi \left(K-k\right)}{K}\right\}+\beta_{k_*} \nonumber\\
    &= 0.
\end{aligned}
\end{equation}
This when combined with the properties of $p_{11}$ proven above, suffices to show that the ring solution $\P_R = \left(P_{11}, P_{22}\right)\left( w \right) = \left(p_{11}, p_{22} \right)\left(f^{-1} \left(w \right) \right)$ vanishes along the diagonals, $\phi=0$ and $\phi=\frac{\pi}{2}$, for a square $E_4$. For $K \neq 4$, we have $P_{11} \neq 0$ for $w\neq \left(0,0\right)$ and hence the origin is the unique zero of the associated ring solution.
\end{proof}
Remark: We briefly remark on the equivalence of $\P_R$ for $E_4$ and the WORS analysed in \cite{canevari2017order}. The WORS is defined in a square domain with edges parallel to the $x$ and $y$-axis respectively, and hence, the eigenvectors are $\x$, $\y$ and $\z$ respectively. The WORS belongs to a class of LdG equilibria of the form
\[\Q = q_1\left(\x\otimes \x - \y\otimes \y\right) + q_2\left(\x\otimes \y + \y\otimes \x\right) - \frac{B}{6C}\left(2\z\otimes \z - \x\otimes \x - \y\otimes \y \right)\]
at $A = -\frac{B^2}{3C}$, and the WORS has $q_2$ identically zero everywhere. In Proposition~\ref{isotropic_center}, we rotate the square by $45$ degrees, so that $(q_1, q_2)$ are related to $\P_R$ by
\begin{equation}
    \left(\begin{tabular}{cc}
        $q_1$ & $q_2$\\
        $q_2$ & $-q_1$
    \end{tabular}\right)(\r) = S\P_R(S^T\r)S^T =
    \left(\begin{tabular}{cc}
        $-P_{12}$ & $P_{11}$\\
        $P_{11}$ & $P_{12}$
    \end{tabular}\right)
    (S^T\r)
\end{equation}
where $S$ is the corresponding rotation matrix. Hence,
$q_2 = 0$ in \cite{canevari2017order} translates to $P_{11}=0$ in Proposition~\ref{isotropic_center}.\\
\begin{figure}
    \begin{center}
        \includegraphics[width=0.9\columnwidth]{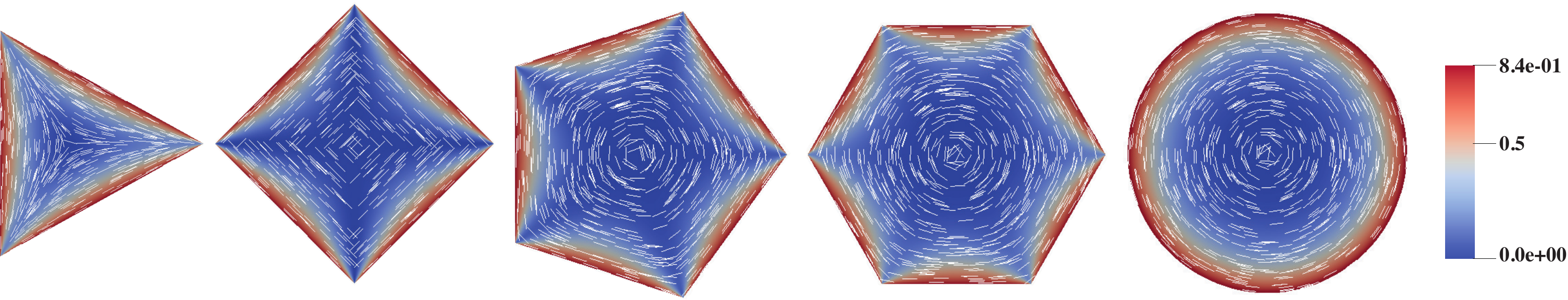}
        \caption{Solutions $\left(P_{11}^0,P_{12}^0\right)$ of (\ref{P11_P12_laplace}) when $K = 3,4,5,6$ in regular triangle, square, pentagon, hexagon domain and $K\to\infty$ in disc domain.
        The vector $\left(\cos\left(arctan\left(P_{12}^0/P_{11}^0\right)/2\right),\sin\left(\arctan\left(P_{12}^0/P_{11}^0\right)/2\right)\right)$ is represented by white lines and the order parameter $\left(s^0\right)^2=\left(P_{11}^0\right)^2 + \left(P_{12}^0\right)^2$ is represented by color from blue to red. The maximum of $(s^0)^2$ on boundary is $\left(\frac{B}{2C}\right)^2 \approx 0.84$, with constant $B = 0.64\times10^4 N/m^2$ and $C = 0.35\times 10^4 N/m^2$.}
        \label{direction_multi_edge}
    \end{center}
\end{figure}

With Proposition~\ref{isotropic_center}, we address the question raised at the beginning of this section. The Ring solution, $\P_R$, is the unique solution of the limiting problem (\ref{zero_euler}) and provides an excellent approximation to global minima of the reduced energy (\ref{p_energy}) for $\lambda$ sufficiently small, for all $E_K$ with $K\geq 3$ (see error estimates in (\ref{eq:error_estimates})). The square, $E_4$ is special since the eigenvectors of the associated $\P_R$ are constant in space and $\P_R$ vanishes along the square diagonals. For $K\neq 4$, $\P_R$ has a unique isotropic point at the origin and is referred to as the ring solution, since for $K>4$, the director profile (the profile of the leading eigenvector of $\P_R$ with the largest positive eigenvalue) follows the profile of a $+1$-vortex located at the centre of the polygon. In Figure \ref{direction_multi_edge}, we numerically plot the ring configuration for a triangle, square, pentagon, hexagon and a disc.
For $K=3$, the isotropic point at the centre of the equilateral triangle resembles a $-1/2$ nematic point defect. This is a very interesting example of the effect of geometry on solutions with profound optical and experimental implications.

Following Lemma $6.1$ in \cite{canevari2017order}, we can prove that for any $\lambda>0$, there exists a critical point $\P_s\in C^2\left(E_K\right)\cap C^0\left(\overline{E_K}\right)$ of (\ref{p_energy}) which satisfies the boundary condition $\P_s = \P_b$ on $\partial E_K$, in the class
$\mathscr{A}_{sym}=\{\P\in\mathscr{A};\P(\r) = S\P(S^T\r)S^T,S\in G_K\}$,
where $G_K = \{S\in O(2):SE_K\in E_K\}$, and $\P_s\left(0,0\right) = 0$.
\begin{figure}
    \begin{center}
        \includegraphics[width=0.6\columnwidth]{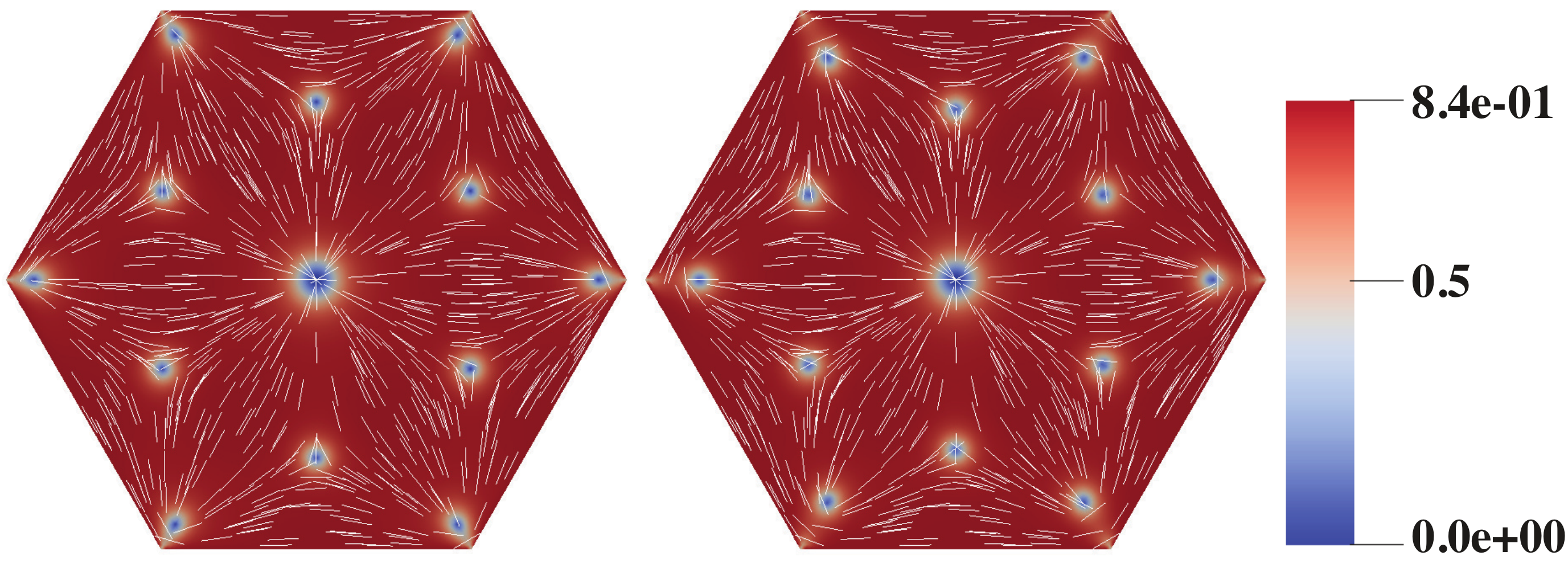}
        \caption{Two symmetric critical points of (\ref{p_energy})
        with multiple interior zeros when $\lambda^2 = 1500$.}
        \label{fig:interior_zeroes}
    \end{center}
\end{figure}
We refer to these critical points as ``symmetric critical points". The ring solution, $\P_R$ is a special example of a symmetric critical point at $\lambda=0$. However, we numerically find symmetric critical points with the zero at the origin and multiple interior zeroes, as illustrated on a hexagon, $E_6$ in Figure~\ref{fig:interior_zeroes}. These critical points, $\P_c$, with multiple zeroes are unstable critical points of (\ref{p_energy}) in the sense that the associated second variation of the reduced energy
\begin{equation}
    \partial^2F_{\lambda}[\eta] = \int_{E_K}|\nabla \eta|^2 + \frac{\lambda^2}{4}\left(|\P_c|^2-\frac{B^2}{2C^2}\right)|\eta|^2+\frac{\lambda^2}{2}\left(\P_c\cdot\eta\right)^2
	\label{eq:second}
\end{equation}
 has negative eigenvalue, where $\eta$ is an arbitrary symmetric, traceless $2\times 2$ matrix vanishing on $\partial E_K$. In fact, in \cite{canevari2017order}, the authors prove that for the WORS, the smallest eigenvalue of (\ref{eq:second}) is strictly decreasing with $\lambda$. We refer to the unique minimizer of (\ref{p_energy}) for sufficiently small $\lambda$ as being ``ring-like" since they are uniformly close to $\P_R$ from the error estimates in (\ref{eq:error_estimates}). By analogy with the work in \cite{canevari2017order}, we expect the smallest eigenvalue of the second variation of the reduced energy in (\ref{eq:second}) about the ring-like solutions, to be a decreasing function of $\lambda$, so that the ring-like solution branch is globally stable for small $\lambda$ and is unstable for large $\lambda$.

Whilst $\P_R$ has been discussed in a strictly two-dimensional setting, it is worth pointing out the 3D relevance of the ring solution. In \cite{canevari_majumdar_wang_harris}, the authors prove that the WORS is the global LdG energy minimizer on three-dimensional wells with a square cross-section, for $\lambda$ sufficiently small and for all choices of the well height, with at least two different choices of boundary conditions on the top and bottom surfaces of the well. The same remarks apply to the ring solution, $\P_R$, for three-dimensional wells that have $E_K$ as their cross-section. In other words, $\P_R$ is a physically relevant approximation to global LdG minima on three-dimensional wells with a regular polygon cross-section, for $\lambda$ sufficiently small, independently of well height. Further, as $\lambda$ increases, the authors report novel mixed solutions on three-dimensional wells with a square cross-section that exhibit the WORS profile at the centre of the well. Using similar reasoning, we expect ring-like solutions to lose stability as $\lambda$ increases on 3D wells, that have $E_K$ as their cross-section. However, they may be observable in mixed solutions, making them of relevance in the large $\lambda$-regime too.
Finally, we numerically check how well $\P_R$ approximates solutions of the nonlinear system (\ref{Euler_Lagrange}) for small $\lambda$.
We use FEniCS package~\cite{olgg2012fenics} to solve the  Laplace equation for $\P_R$ with Dirichlet boundary conditions. We set the boundary value at the vertices to be the average of the two constant values on the intersecting edges at the vertex in question. We use standard FEM (Finite Element Methods) and the Newton's method to solve the nonlinear system (\ref{Euler_Lagrange}) for small $\lambda$.
\begin{figure}
\centering
        \includegraphics[width=1\columnwidth]{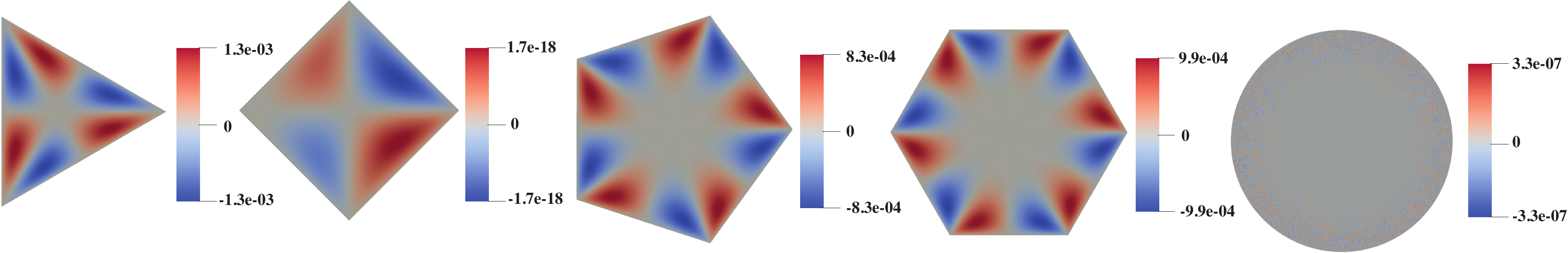}
        \caption{$P_{11}^1P_{12}^0-P_{12}^1P_{11}^0 = s^0s^1\sin\left(2\gamma^0-2\gamma^1\right)$ for regular triangle, square, pentagon, hexagon and disc.}
        \label{tan2theta_1_error}
\end{figure}
In Figure \ref{tan2theta_1_error}, we consider $\P^1$ as the numerical solution of (\ref{Euler_Lagrange}) with $\lambda^2 = 1$ and $\P^0$ as the numerically computed ring solution with $\lambda^2 = 0$. In Figure \ref{tan2theta_1_error}, we plot
$P_{11}^1P_{12}^0-P_{12}^1P_{11}^0 = s^0s^1\sin\left(2\gamma^0-2\gamma^1\right)$ for a regular triangle, square, pentagon, hexagon and disc respectively, where $\left(P_{11}^0, P_{12}^0 \right) = s^0 \left(\cos 2 \gamma^0, \sin 2 \gamma^0 \right)$ and $\left(P_{11}^1, P_{12}^1 \right) = s^1 \left(\cos 2 \gamma^1, \sin 2 \gamma^1 \right)$.
The color bars show that the maximum difference for a triangle, pentagon and hexagon is about $1e-3$, however the difference for square and disc is much lower, $1.7e-18$ and $3.3e-7$ respectively.
This is simply because the eigenvectors of $\P^1$ and $\P^0$ are the same on a square and a disc i.e. for a square, the eigenvectors are $\x$ and $\y$ respectively whereas the eigenvectors are the radial unit-vector and the azimuthal unit-vector on a disc for any $\lambda$\cite{canevari2017order,ignat2016instability}.
The eigenvectors do change with $\lambda$ on $E_K$ for $K\neq 4$ and this explains the larger error for $K\neq 4$ noted above.
\begin{figure}
\centering
        \includegraphics[width=1\columnwidth]{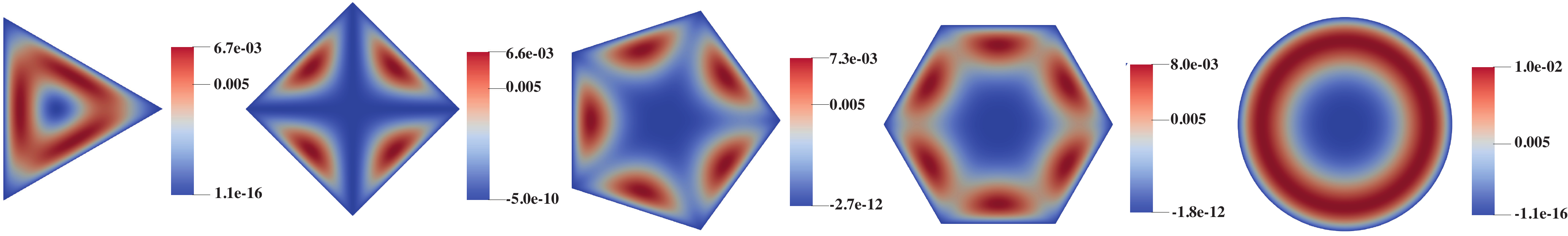}
        \caption{$|P^1|^2/2-|P^0|^2/2 = \left(s^1\right)^2-\left(s^0\right)^2$ for regular triangle, square, pentagon, hexagon and disc.}
        \label{s_1_error}
\end{figure}
We also plot $(s^1)^2-(s^0)^2$ for a regular triangle, square, pentagon, hexagon and disc
in Figure \ref{s_1_error} and the differences are within $1e-2$. These numerical experiments demonstrate the validity of $\P_R$ as an excellent approximation to minima of (\ref{p_energy}) for small $\lambda$.
\begin{figure}
\centering
        \includegraphics[width=0.6\columnwidth]{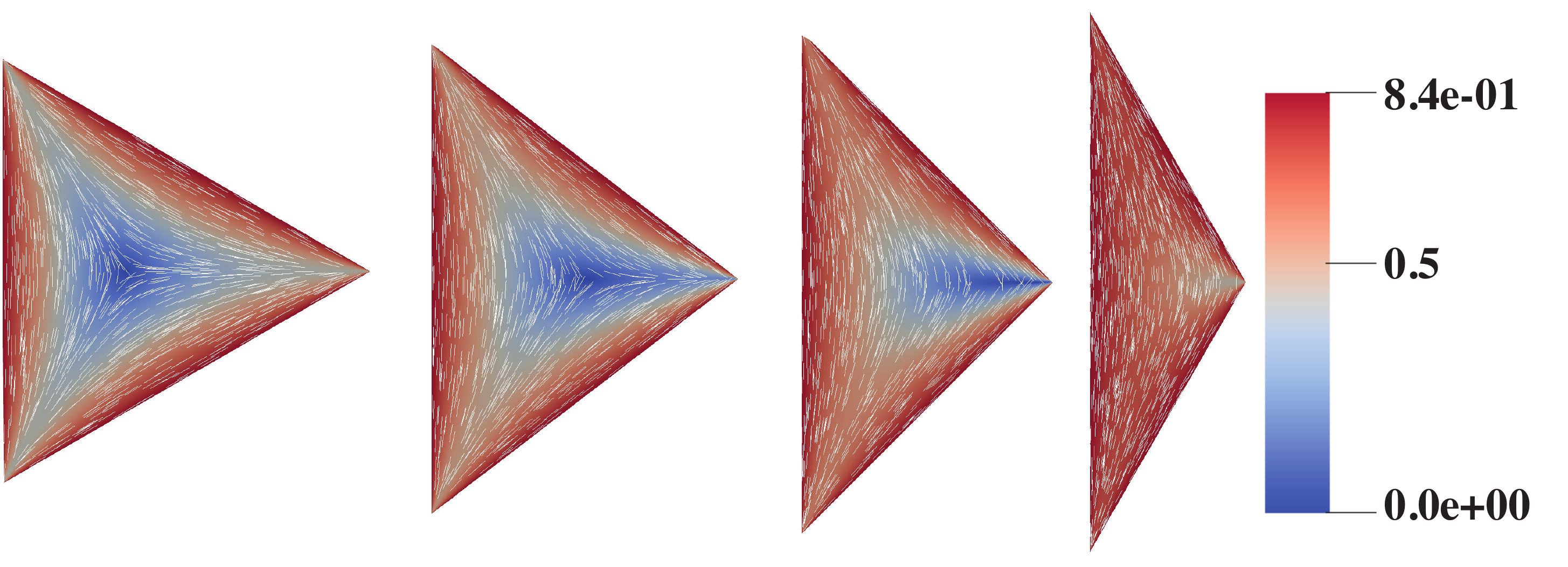}
        \caption{The solutions of (\ref{P11_P12_laplace}) with corresponding tangential boundary condition in isosceles triangles domain with the top angle $120^0$, $90^0$, $75^0$ and $60^0$ respectively.
        The vector $\left(\cos\left(\arctan\left(P_{12}^0/P_{11}^0\right)/2\right),\sin\left(\arctan\left(P_{12}^0/P_{11}^0\right)/2\right)\right)$ is represented by white lines and the order parameter $\left(s^0\right)^2=\left(P_{11}^0\right)^2 + \left(P_{12}^0\right)^2$ is represented by color from blue to red.}
        \label{nonregular_triangle}
\end{figure}
Finally, in Figure \ref{nonregular_triangle}, we numerically compute the solution of the Laplace boundary value problem for the matrix $\P$, on different isosceles triangles subject to Dirichlet tangent boundary conditions. We numerically observe a single isotropic point migrating from the apex vertex to the centre of the triangle, as the angle at the apex decreases from $120^0$ to $60^0$ ($E_3$). This again illustrates the effect of geometry on the location of the isotropic points/optical singularities.
\subsection{The $\lambda\rightarrow\infty$ Limit or the Oseen-Frank Limit}
\label{subsec:infty_limit}
\subsubsection{The Number of Stable States}
The $\lambda\to\infty$ limit is analogous to the ``vanishing elastic constant limit" or the ``Oseen-Frank limit" in \cite{majumdar2010landau}. Let $\P^{\lambda}$ be a global minimizer of (\ref{p_energy}), subject to a fixed boundary condition $\left(P_{11b},P_{12b}\right)$ on $\partial E_K$. As $\lambda \to \infty$, the minima, $\P^{\lambda}$, converge strongly in $W^{1,2}$ to $\P^{\infty}$ where
    \begin{equation}
        \P^{\infty} = \frac{B}{2C}\left(\n^{\infty}\otimes \n^{\infty}-\frac{1}{2}\I_2\right),\nonumber
    \end{equation}
    $\n^{\infty} = \left(\cos\gamma^{\infty},\sin\gamma^{\infty}\right)$ and $\gamma^{\infty}$ is a global minimizer of the energy
    \[
        I[\gamma]: = \int_{E_K} \left| \nabla \gamma \right|^2 \mathrm{dA}
    \]
subject to Dirichlet conditions, $\gamma=\gamma_b$ on $\partial E_K$. Setting $\n_b = \left(\cos\gamma_b, \sin\gamma_b \right)$, we have $\n_b$ is tangent to the edges $C_k$, which constrains the values of $\gamma_b$ on $C_k$, and if $\textrm{deg}\left(\n_b,\partial E_K\right) = 0$, then $\gamma^{\infty}$ is a solution of the Laplace equation
    \begin{equation}
        \begin{aligned}
            \Delta\gamma^{\infty} &= 0,\ on\ E_K\\
        \end{aligned}
        \label{infty_euler}
    \end{equation}
subject to $\gamma=\gamma_b$ on $\partial E_K$ \cite{lewis2014colloidal,bethuel1993asymptotics}.
Since we are largely presenting heuristic arguments in this section, we take $\gamma_b$ to be piecewise constant on the edges $C_k$, consistent with the tangent conditions for $\n_b$ on $\partial E_K$. This choice of $\gamma_b$ would not work for the Dirichlet energy due to the discontinuities at the corners \cite{lewis2014colloidal}.

There are multiple choices of Dirichlet conditions for $\gamma_b$ consistent with the tangent boundary conditions, which implies that there are multiple local/global minima of (\ref{p_energy}) for large $\lambda$. We present a simple estimate of the number of stable states if we restrict $\gamma_b$ so that $\gamma^\infty$ rotates by either $2\pi/K-\pi$ or $2\pi/K$ at a vertex (see Figure \ref{normal_angle}(a) and (b), referred to as ``splay" and ``bend" vertices respectively). Since we require $\textrm{deg}\left(\n_b,\partial E_K\right) = 0$, we necessarily have $x$ ``splay" vertices and $\left( K - x \right)$ ``bend" vertices such that
\begin{equation}
    x\left(2\pi/K-\pi\right)+\left(K-x\right)\left(2\pi/K\right) = 0\nonumber
\end{equation}
 only when $x = 2$. We thus have $(K-2)$ bend corners and $2$ splay corners. We can define a topological charge with each corner, associated with the amount of director rotation about the corner. Skipping the technical details, a bend corner has winding number $w_b = -\frac{2\pi}{K}\div 2\pi= -\frac{1}{K}$ and a splay corner has winding number $w_S = \frac{(K-2)\pi}{K}\div 2\pi = \frac{K-2}{2K}$. The total winding number is zero. This is consistent with the results in \cite{yao2018topological}, where the authors claim that the general rule of the total winding number of a 2D liquid crystal in a polygon with $K$ sides is $-\frac{K-2}{2}$ under the assumption that molecules always make a splay pattern at the polygon corners.
So we have at least $K\choose 2$ minima of (\ref{p_energy}) for $\lambda$ sufficiently large.
\begin{figure}
    \begin{center}
        \includegraphics[width=0.4\columnwidth]{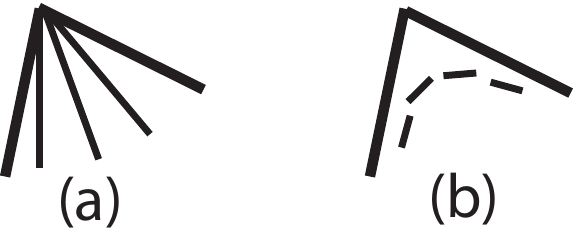}
        \caption{Two arrangements of nematics in the corner: (a) splay and (b) bend}
        \label{normal_angle}
    \end{center}
\end{figure}
As an illustrative example, we take the hexagon $E_6$  in Figure \ref{C62_15_states}.
The Dirichlet boundary conditions are
\begin{equation}
    \gamma_b = \gamma_k\ on\ C_k,\ k = 1,...,K,
    \label{infty_boundary}
\end{equation}
where
\begin{equation}
    \gamma_1 = \frac{\pi}{K}-\frac{\pi}{2},\ \gamma_{k+1} = \gamma_k+jump_k,\ k = 1,2,..,K-1.\nonumber
\end{equation}
We need to choose the two splay vertices where $\gamma$ rotates as in Figure \ref{normal_angle}(a). If the chosen corner is between $C_k$ and $C_{k+1}$, then
$jump_k = 2\pi/K-\pi$, otherwise $jump_k = 2\pi/K$, $k = 1,...,K-1$. We have $15$ different choices for the two ``splay" vertices, (i) $3$ of which correspond to the three pairs of diagonally opposite vertices, (ii) $6$ of which correspond to pairs of vertices which are separated by one vertex and (iii) $6$ of which correspond to ``adjacent" vertices connected by an edge (see Figure \ref{C62_15_states}). We refer to (i) as \emph{Para} states, (ii) as \emph{Meta} states and (iii) as \emph{Ortho} states. All $15$ states are locally stable in the sense that the corresponding second variation of (\ref{p_energy}) (see (\ref{eq:second})) is strictly positive according to our numerical computations.

\begin{figure}
    \begin{center}
        \includegraphics[width=0.9\columnwidth]{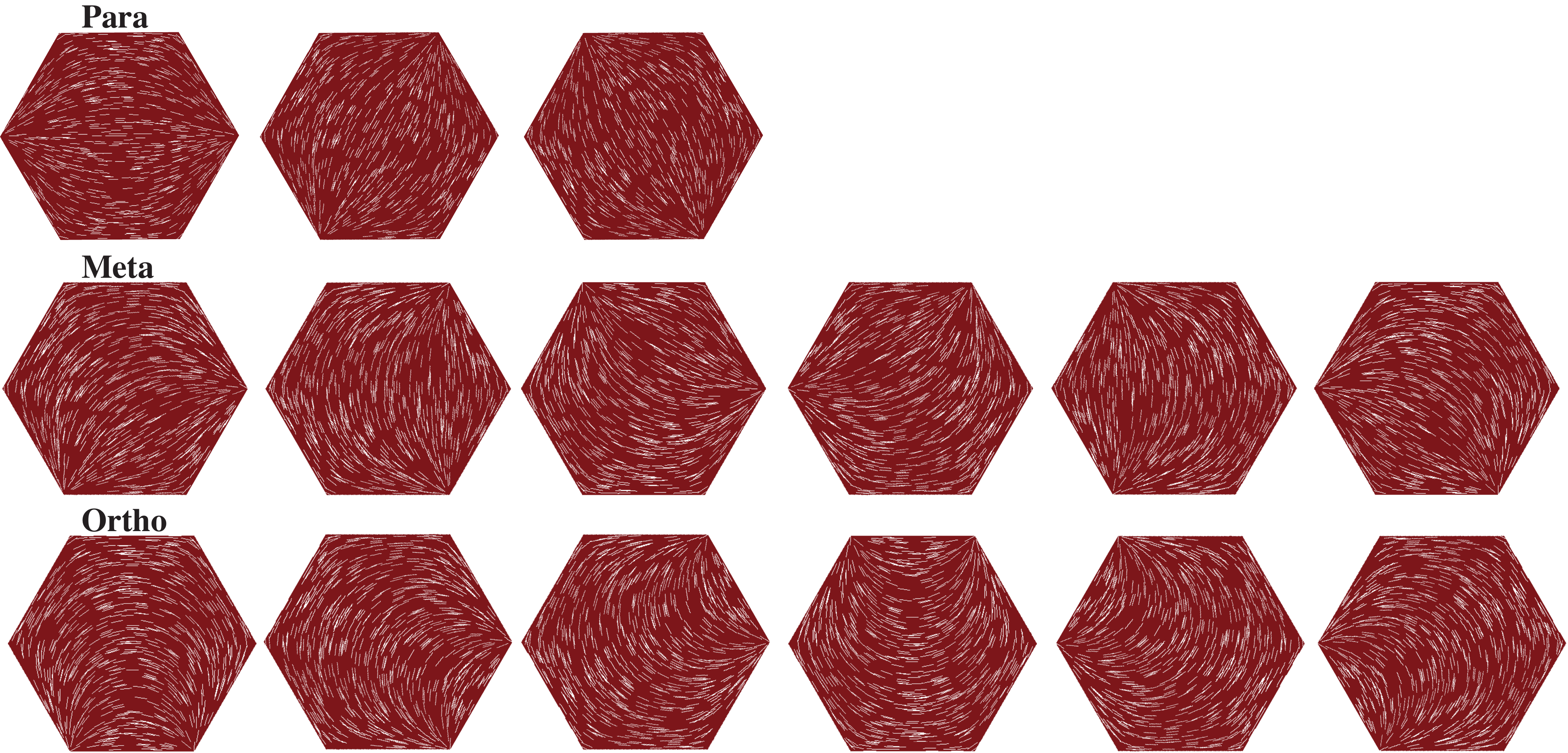}
        \caption{$6\choose 2 = 15$ solutions of (\ref{infty_euler}) subject to boundary condition (\ref{infty_boundary}) in hexagon domain.
        The vector $\left(\cos\gamma^{\infty},\sin\gamma^{\infty}\right)$ is represented by white lines. The red color indicates the order parameter $s^{\infty}=\frac{B}{2C}$, in order to facilitate comparison with the solution in Figure \ref{s_theta_2250_hexagon}.}
        \label{C62_15_states}
    \end{center}
\end{figure}
\subsubsection{The limiting profiles in (\ref{infty_euler}) are good approximations to solutions of (\ref{Euler_Lagrange}) for large $\lambda$}

In the numerical simulations, we take $B = 0.64\times 10^{4} N/m^2$ and $C = 0.35\times10^4 N/m^2$ to be fixed constants (also see \cite{canevari2017order}). In particular, this choice dictates the boundary values for $P_{11}$ and $P_{12}$ on $\partial E_K$.
    For large $\lambda$, the defect core sizes are very small and we have an intrinsic multi-scale problem. The limiting problem (\ref{infty_euler}) has no length scale and in what follows, we compare the limiting profiles in (\ref{infty_euler}) with solutions of (\ref{Euler_Lagrange}) for large but numerically tractable values of $\lambda$.
We take the regular hexagon as an example. For $\lambda^2 = 2250$, we compute three distinct Para, Meta and Ortho solutions of (\ref{Euler_Lagrange}) with different initial conditions.
We label the solutions as $\left(P_{11}^{2250},P_{12}^{2250}\right) = s^{2250}\left(\cos2\gamma^{2250},\sin2\gamma^{2250}\right)$. Similarly, we compute  $\left(P_{11}^{\infty},P_{12}^{\infty}\right) = s^{\infty}\left(\cos2\gamma^{\infty},\sin2\gamma^{\infty}\right)$, where $\gamma^{\infty}$ is the unique solution in (\ref{infty_euler}) subject to a fixed boundary condition and $s^{\infty} \equiv \frac{B}{2C}$. For three different choices of the boundary conditions, we numerically compute three different solutions, $\gamma^\infty_P$, $\gamma^\infty_M$ and $\gamma^\infty_O$, where $P, M,O$ label Para, Meta and Ortho respectively. The three different solutions for $\gamma^\infty$ yield the corresponding Para, Meta and Ortho profiles for $\P^\infty$ respectively. In all three cases, we numerically compute the measure $P_{11}^{2250}P_{12}^{\infty}-P_{12}^{2250}P_{11}^{\infty}$ and see that the measure concentrates near the pairs of splay vertices. Analogously, the measure, $|\P^{\infty}|^2-|\P^{2250}|^2$, also concentrates at the splay vertices i.e. $s^{2250}$ drops at the splay vertices (so these can be interpreted as localised defects where $\n_b$ has a discontinuity which cannot be removed by smoothening the corners of $E_K$) whereas $s^\infty$ is fixed (more details are visible in Figure \ref{s_theta_2250_hexagon}). We deduce that $\P^\infty$ is a good approximation to $\P^\lambda$ for $ \lambda$ sufficiently large, since the maximum numerical error is $10^{-4}$ away from the splay vertices. We do not have asymptotic expansions for $\P^{\lambda}$ to ascertain convergence rates at hand and this will be pursued in future work.
\begin{figure}
\centering
        \includegraphics[width=0.7\columnwidth]{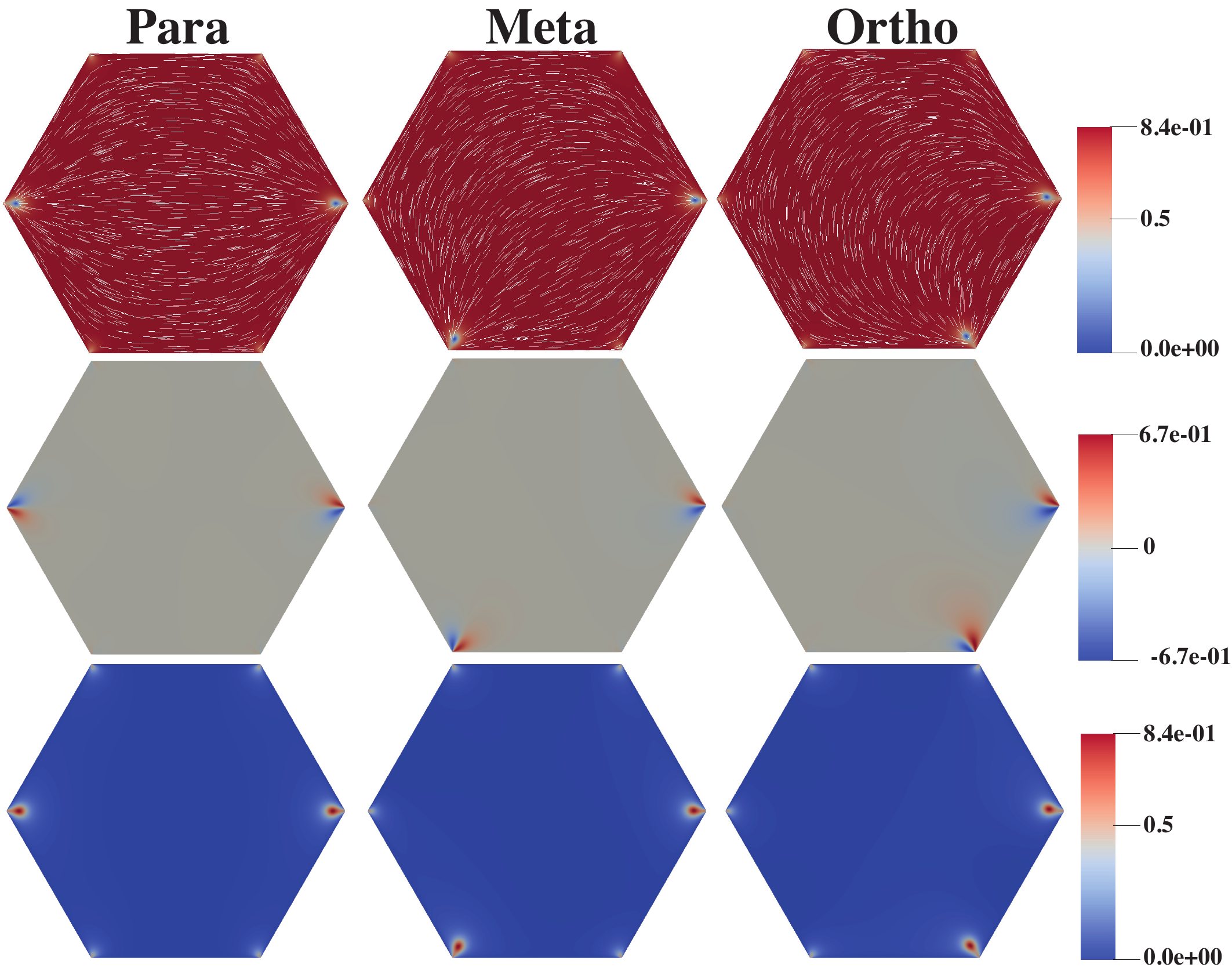}
        \caption{The images in the first row show the Ortho, Meta and Para solutions of (\ref{Euler_Lagrange}) with $\lambda^2 = 2250$. The images in the second and third rows show $P_{11}^{2250}P_{12}^{\infty}-P_{12}^{2250}P_{11}^{\infty} = s^{2250}s^{\infty}\sin\left(2\gamma^{\infty}-2\gamma^{2250}\right)$ and $\left(\P^{\infty}\right)^2/2-\left(\P^{2250}\right)^2/2 = \left(s^{\infty}\right)^2-\left(s^{2250}\right)^2$, respectively.}
        \label{s_theta_2250_hexagon}
\end{figure}

\subsubsection{Numerical methods}
\label{sec:nummethods}
We use the weak formulation of (\ref{Euler_Lagrange}) given by
\begin{equation}
\begin{aligned}
    0=&\int_{\Omega}\nabla P_{11}\cdot\nabla v_{11}+ \lambda^2\left(P_{11}^2+P_{12}^2-\frac{B^2}{4C^2}\right)P_{11}v_{11} \mathrm{dA},\\
    0=&\int_{\Omega}\nabla P_{12}\cdot\nabla v_{12}+ \lambda^2\left(P_{11}^2+P_{12}^2-\frac{B^2}{4C^2}\right)P_{12}v_{12} \mathrm{dA},
\end{aligned}
    \label{weak_lambda}
\end{equation}
to numerically compute the critical points of (\ref{p_energy}) for $0<\lambda<\infty$,
where $v_{11}$, $v_{12}$ are arbitrary test functions.
We use a triangle mesh for the domain, with mesh-size $h\leq\frac{1}{256}$, and the mesh is fixed in the numerical simulations. We set the value at the polygon vertices to be the average of the constant values on the two intersecting edges at the vertex in question (as previously mentioned) and provided $\epsilon <  h$ (recall $\epsilon$ is the width of the interpolation interval), we can numerically work with piecewise constant boundary conditions on the edges, $C_K$. Lagrange elements of order 1 are used for the spatial discretization. The linear systems for the limiting cases, $\lambda=0$ and $\lambda \to \infty$, are solved using LU solver and the nonlinear system in (\ref{weak_lambda}) is solved using a Newton solver, with a linear LU solver at each iteration. The tolerance is set to $1e-13$. Newton's method strongly depends on the initial condition and to obtain Ring-like solutions for small $\lambda$, we simply use $\P_R$ as the initial condition.
For large $\lambda$ and for the case of $E_6$, we choose $15$ different $\gamma_b$'s in (\ref{infty_boundary}) to compute the Para, Meta and Ortho states and use these limiting profiles, $\P^\infty$, as initial conditions for (\ref{weak_lambda}), for sufficiently large $\lambda$.

 We perform an increasing $\lambda$ sweep for the Ring branch and decreasing $\lambda$ sweep for distinct Para, Meta or Ortho solution branches to compute the bifurcation diagrams.
Once we obtain the solutions, we numerically compute their free energies by
\begin{equation}
    F[P_{11},P_{12}]: = \int_{\Omega}|\nabla P_{11}|^2+|\nabla P_{12}|^2+\frac{\lambda^2}{2}\left(P_{11}^2+P_{12}^2-\frac{B^2}{4C^2}\right)^2\mathrm{dA},
        \label{positive_energy}
\end{equation}
which is equivalent to (\ref{p_energy}), modulo a constant.
In this paper, all finite-element simulations and numerical integrations are performed using the open-source package FEniCS~\cite{olgg2012fenics}.
We study the stability of the solutions of (\ref{weak_lambda}) by numerically calculating the smallest real eigenvalue of the Hessian of the reduced energy (\ref{p_energy}) and the corresponding eigenfunction using the LOBPCG (locally optimal block preconditioned conjugate gradient) method in \cite{yin2019high,yin2019cons} (which is an iterative algorithm to find the smallest (largest) $k$ eigenvalues of a real symmetric matrix.) A negative eigenvalue is a signature of instability and we have local stability if all eigenvalues are positive.
We numerically compute a bifurcation diagram for the critical points of (\ref{p_energy}) on a hexagon and a pentagon in the next section, as a function of the edge length $\lambda$.

\section{Bifurcation Diagram for Reduced LdG Critical Points - Some Examples}
\label{subsec:numerics}
In \cite{robinson2017molecular}, the authors extensively discuss the reduced LdG bifurcation diagram on a square domain, as a function of the square length $D$.
For $D$ small enough, the WORS with an isotropic cross along the square diagonals, as shown in Figure \ref{direction_multi_edge}, is the unique solution. There is a bifurcation point at $D = D^*$ such that WORS is stable for $D<D^*$ and is unstable for $D>D^*$. The WORS bifurcates into stable diagonal solutions, labelled as D1 and D2 solutions, for which the nematic director is aligned along one of the square diagonals. There is a second bifurcation into unstable BD1 and BD2 solutions, which are featured by isotropic lines or defect lines localised near a pair of opposite edges. As $D$ increases further, there is a further critical value, $D = D^{**} > D^{*}$, for which BD1 and BD2 respectively bifurcate into two rotated states, R1, R2 for which the director rotates by $\pi$ radians between a pair of horizontal edges, and R3, R4 solutions, for which the director rotates by $\pi$ radians between a pair of vertical edges. These rotated states gain stability as $D$ increases and for $D\gg D^{**}$, there are six distinct stable solutions: two diagonal and four rotated states. The WORS exists for all $D$ as mentioned above.

Similarly, for a disc of sufficiently small disc radius, the Ring solution with $+1$-defect at the centre, referred to as PR (planar radial), is the unique solution. As the radius increases, the PR solution becomes unstable and bifurcates into a Para type solution, PP (planar polar), with two +1/2 defects which are on the same diameter.

We present two illustrative examples in this section - the critical points of (\ref{p_energy}) on a hexagon and pentagon as a function of $\lambda$. There are more stable solutions than the square and the domains have less symmetry than a disc, so the bifurcation diagrams are more complex. We discuss $E_6$ first. For sufficiently small $\lambda$, there is a unique ring-like minimizer, which is well approximated by $\P_R$ as discussed above (see in Figure \ref{direction_multi_edge} and Lemma 8.2 of Lamy\cite{lamy2014bifurcation}). For large $\lambda$, there are multiple stable solutions, e.g.  Para, Meta and Ortho, in Figure \ref{C62_15_states}. In Figure \ref{bifurcation_diagram}, we use the $\P^\infty$ states discussed above as initial conditions for large $\lambda$ to compute the corresponding $3$ stable $Para$, $6$ stable $Meta$ and six stable $Ortho$ states by continuing the corresponding $\P^\infty$ branches to smaller values of $\lambda$. This is done using standard arc continuation methods; we calculate the smallest eigenvalue of Jacobian of the right-hand side of (\ref{weak_lambda}). If the smallest eigenvalue is larger than $0$, the solution is stable otherwise the solution is unstable. Similarly, we use $\P_R$ as an initial condition for small $\lambda$ to find ring-like solutions for all $\lambda$, which are stable for small $\lambda$ and lose stability as $\lambda$ increases. Besides the ring-like, Para, Meta and Ortho states, we find three unstable BD states which are characterized by two lines of low order ($|\P|^2$) near two edges. In the BD state, the hexagon is separated into three regions by two ``defective low-order lines" such that the corresponding director (eigenvector with largest positive eigenvalue) is approximately constant in each region.

In Figure \ref{bifurcation_diagram}, we plot the free energy of solutions, in (\ref{positive_energy}), as $\lambda$ varies. In Figure \ref{bifurcation_diagram}, we distinguish between the distinct solution branches by defining two new measures, $\int_{\Omega} P_{12}\left(1+x+y\right)dxdy$ and $\int_{\Omega} P_{11}\left(1+x+y\right)dxdy$, and plot these measures versus $\lambda^2$ for the different solutions.
When $\lambda$ is small, the stable ring-like solution is the unique solution. Our numerics show that the ring-like solution (with the unique zero at the polygon center) exists for all $\lambda$ but there is a critical point $\lambda = \lambda^{*}$,  such that the ring-like solution is unstable for $\lambda>\lambda^{*}$ and bifurcates into two kind of branches: stable Para solution branches; unstable BD branches. The unstable BD branches further bifurcate into unstable Meta solutions at $\lambda=\lambda^{**}$. There is a further critical point $\lambda = \lambda^{***}$ at which the Meta solutions gain stability and continue as stable solution branches as $\lambda$ increases. Stable Ortho solutions appear as solution branches for $\lambda$ is large enough. The energy ordering is as follows: the $Para$ states have the lowest energy and the $Ortho$ states are energetically the most expensive, as can be explained on the heuristic grounds that bending between neighbouring vertices is energetically unfavourable.
\begin{figure}
\centering
    \begin{subfigure}{0.55\textwidth}
        \centering
        \includegraphics[width=\columnwidth]{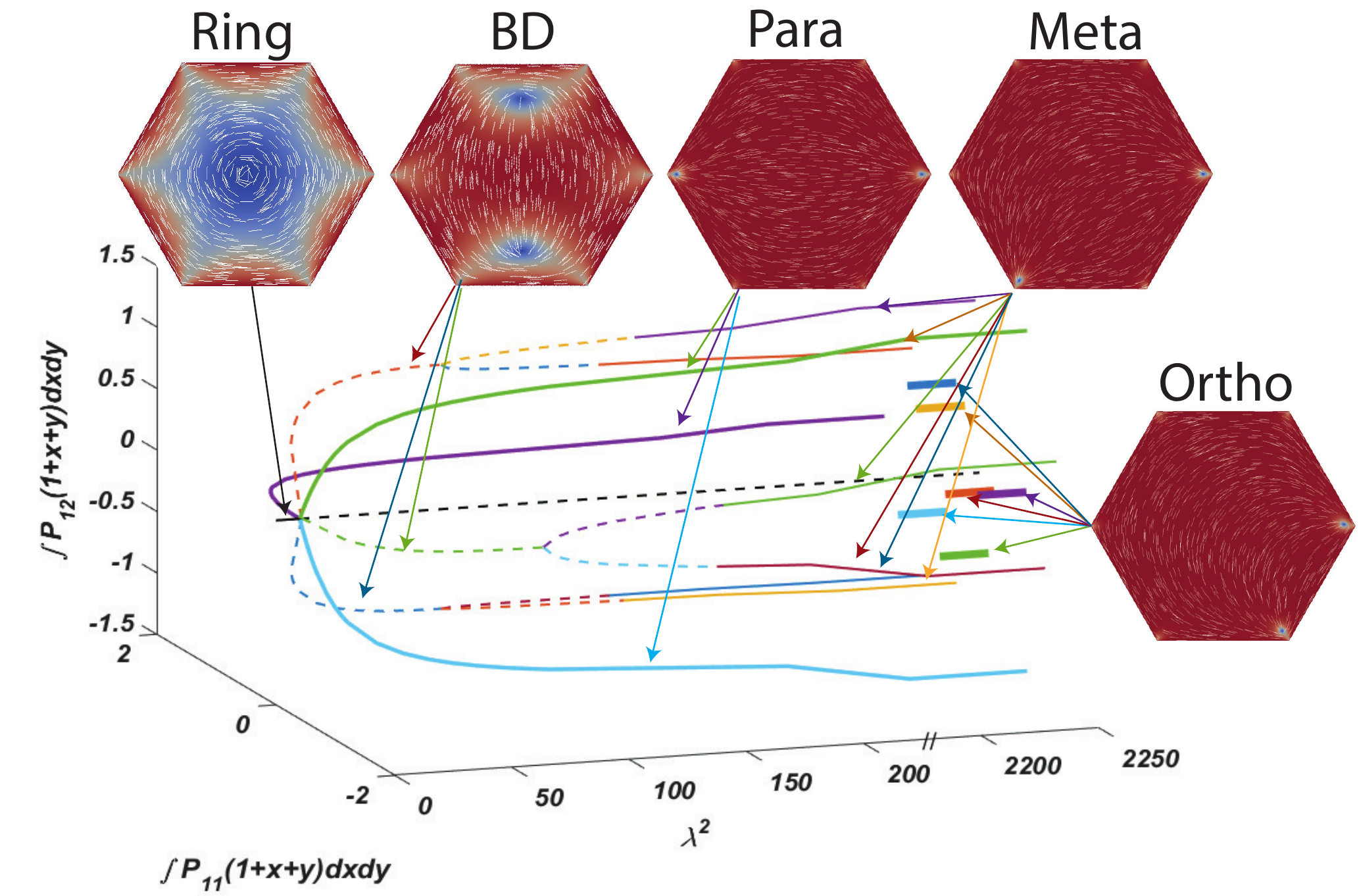}
    \end{subfigure}
    \begin{subfigure}{0.4\textwidth}
        \centering
        \includegraphics[width=\columnwidth]{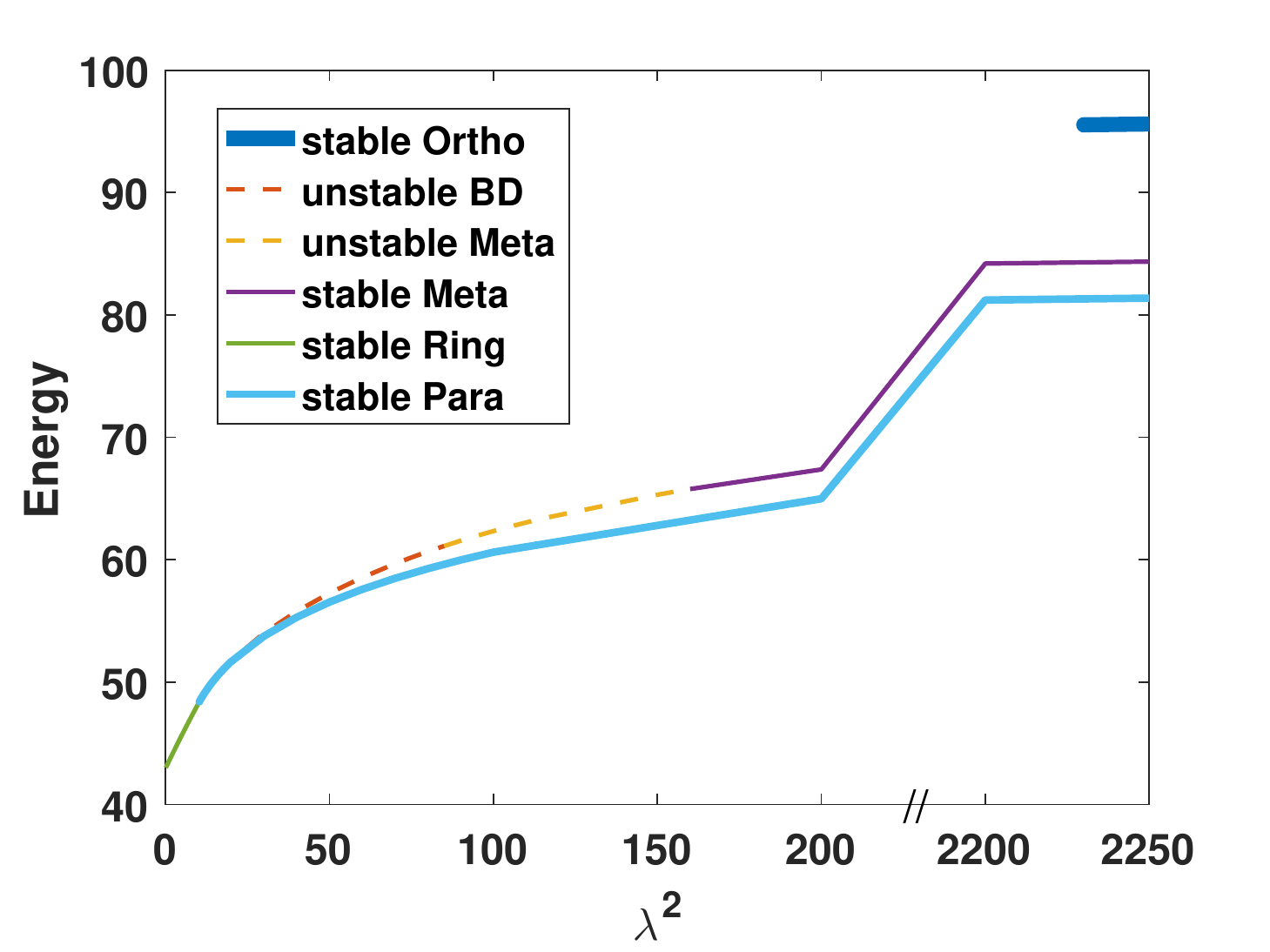}
    \end{subfigure}
    \begin{subfigure}{\textwidth}
        \centering
        \includegraphics[width=0.9\columnwidth]{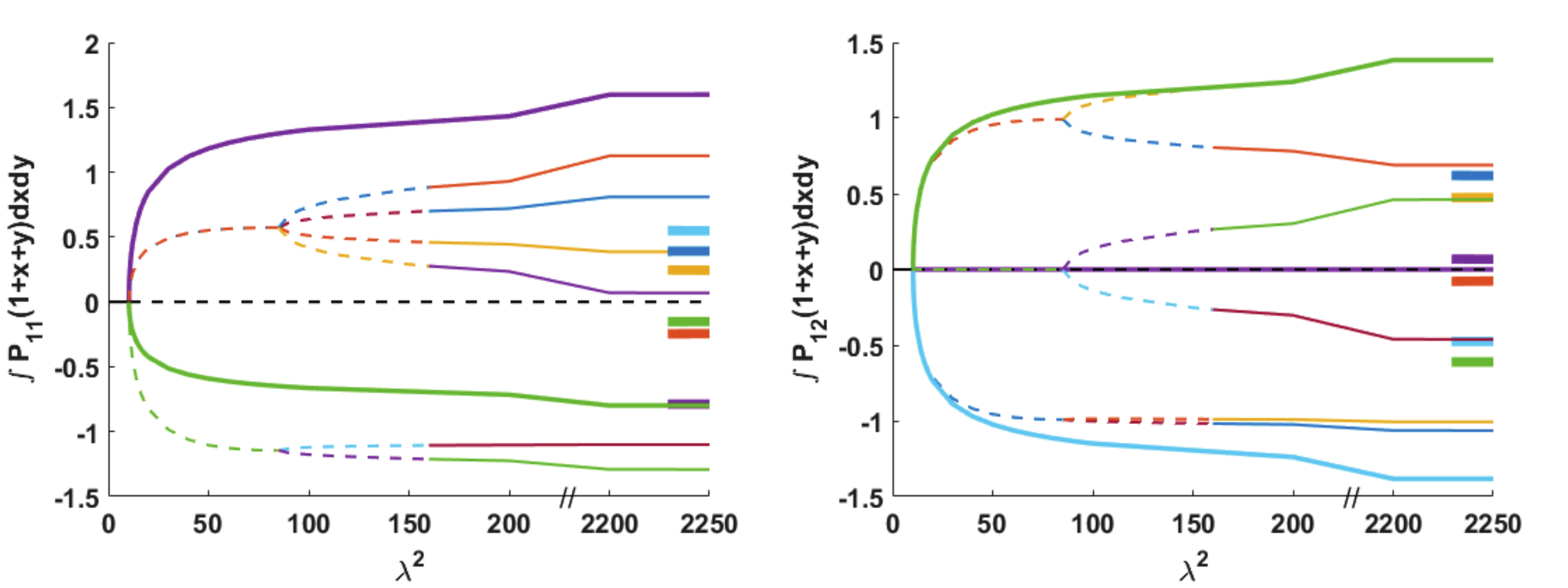}
    \end{subfigure}
        \caption{Bifurcation diagram for reduced LdG model in regular hexagon domain. Top left: plot of $\int P_{11}\left(1+x+y\right)dxdy$, $\int P_{12}\left(1+x+y\right)$ verses $\lambda^2$; top right: plot of the energy in (\ref{positive_energy}) verses $\lambda^2$; bottom: orthogonal 2D projections of the full 3D plot.}
        \label{bifurcation_diagram}
\end{figure}
The case of a pentagon is different. There is no analogue of the $Para$ states and there are $10$ different stable states for large $\lambda$ - (i) five $Meta$ states featured by a pair of splay vertices that are separated by a vertex and (ii) five $Ortho$ states featured by a pair of adjacent splay vertices. There are five analogues of the BD states which are featured by a single line of ``low" order along an edge and an opposite splay vertex. The corresponding bifurcation diagram is illustrated in Figure~\ref{pentagon_bifurcation_diagram}. In all cases, a solid line denotes local stability in the sense of the second variation and a dashed line denotes an unstable critical point.

The examples of a pentagon and a hexagon illustrate some generic features of reduced LdG critical points on polygons with an odd and even number of sides. These examples and the numerical results are not exhaustive but they do showcase the beautiful complexity and ordering transitions feasible in two-dimensional polygonal frameworks.
\begin{figure}
\centering
    \begin{subfigure}{0.55\textwidth}
        \centering
        \includegraphics[width=\columnwidth]{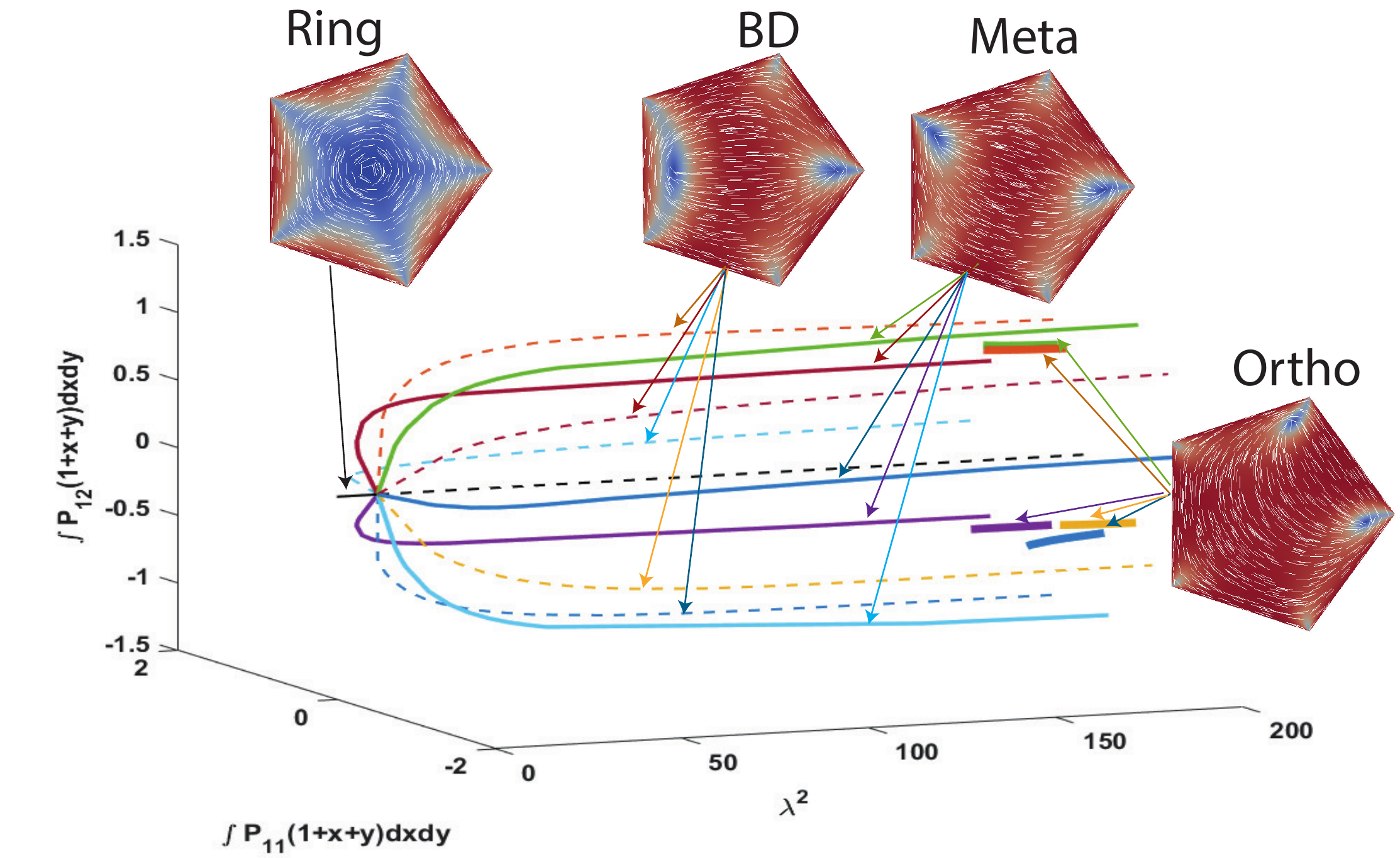}
    \end{subfigure}
    \begin{subfigure}{0.4\textwidth}
        \centering
        \includegraphics[width=\columnwidth]{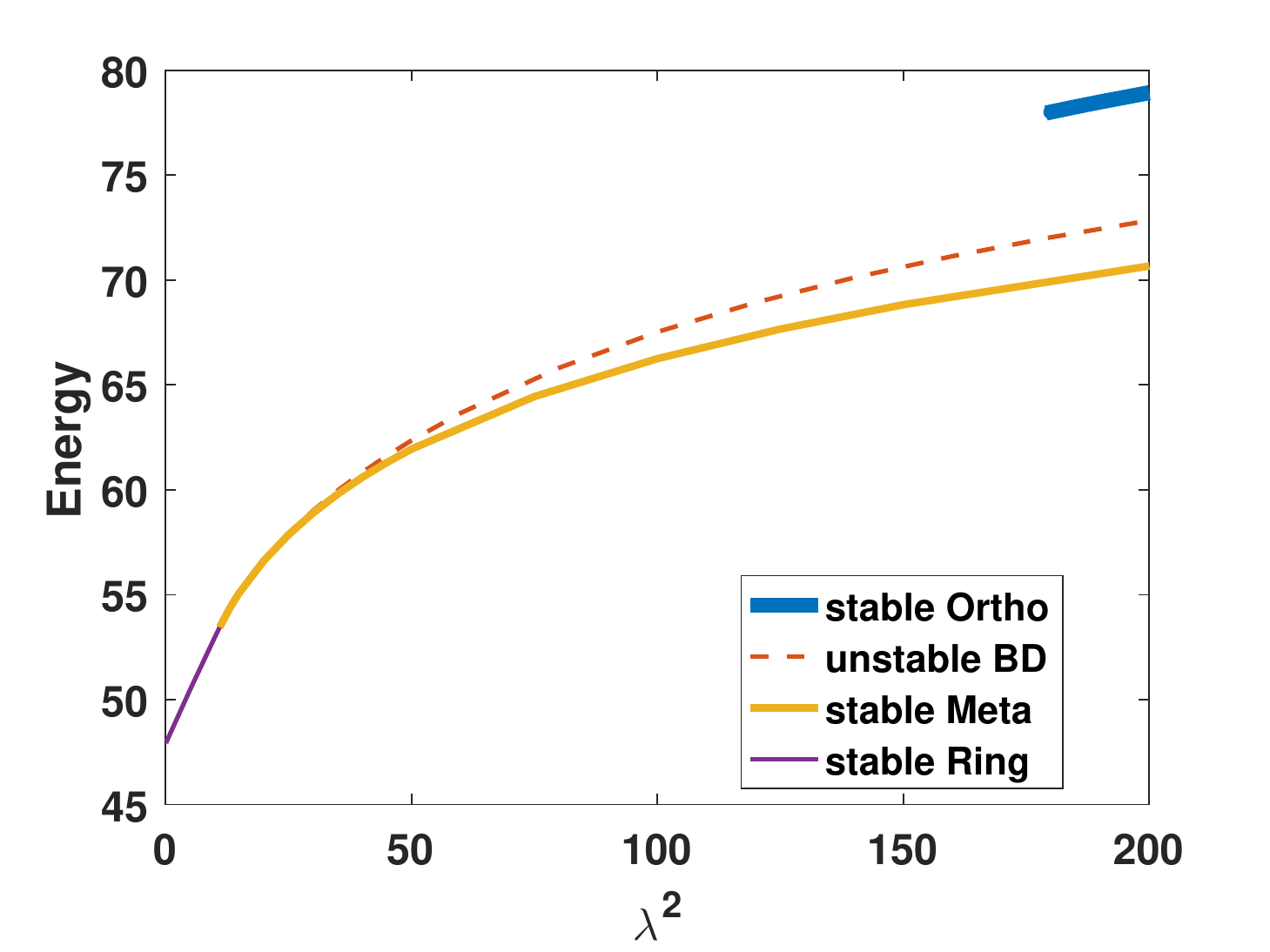}
    \end{subfigure}
    \begin{subfigure}{\textwidth}
        \centering
        \includegraphics[width=0.9\columnwidth]{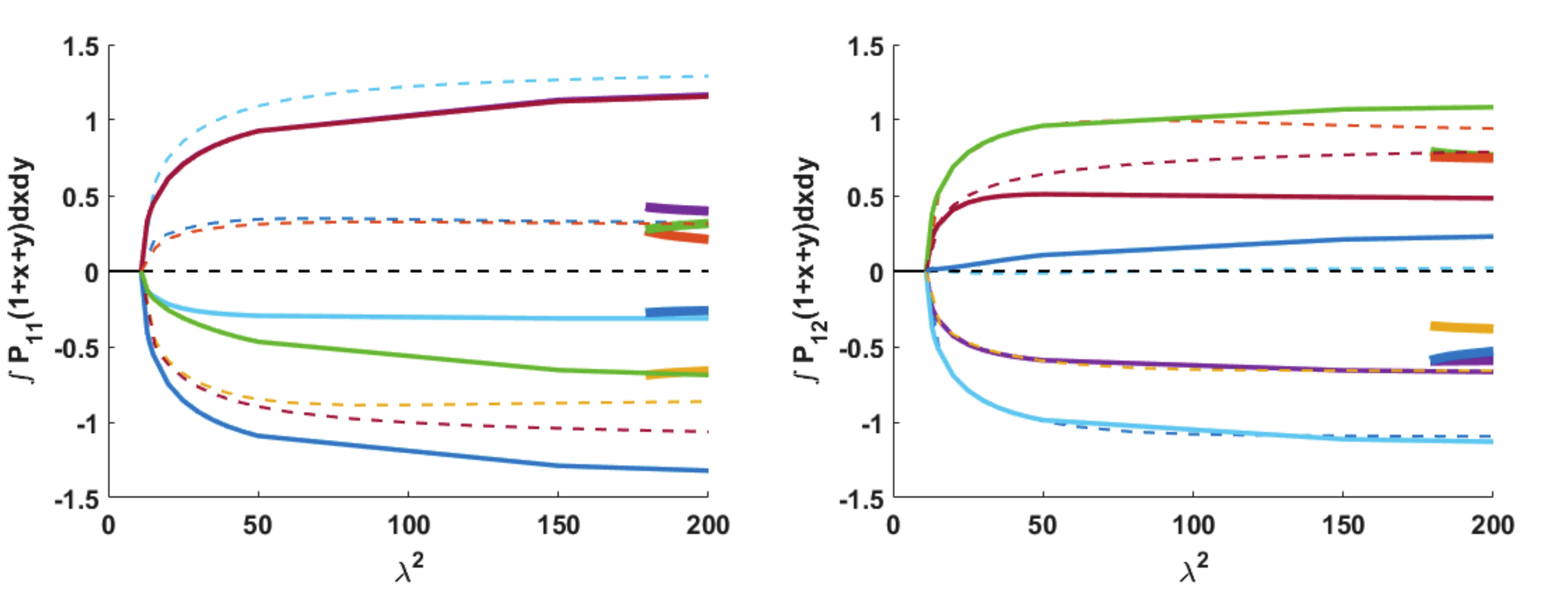}
    \end{subfigure}
    \caption{Bifurcation diagram for reduced LdG model in regular pentagon. Top left: plot of $\int P_{11}\left(1+x+y\right)dxdy$, $\int P_{12}\left(1+x+y\right)$ verses $\lambda^2$; top right: plot of the energy in (\ref{positive_energy}) verses $\lambda^2$; bottom: orthogonal 2D projections of the full 3D plot.}
        \label{pentagon_bifurcation_diagram}
\end{figure}

\section{Conclusion}
\label{sec:conclusion}
We study LdG critical points on 2D regular polygonal domains that have a fixed eigenvector in z-direction, with three degrees of freedom; these critical points are candidates for LdG energy minima in the thin film limit, as established by the Gamma convergence result in \cite{Golovaty2015Dimension}. Further, they also exist in three-dimensional frameworks, e.g. if we work on a well with a regular polygon as cross-section, as illustrated in \cite{canevari_majumdar_wang_harris}.
Working at a fixed temperature, these critical points only have two degrees of freedom and are simply critical points of a rescaled Ginzburg-Landau energy \cite{bethuel1994ginzburg}. Recent work \cite{canevari_majumdar_wang_harris} shows that the qualitative analytic features can be generalised to all temperatures $A<0$, at least in the case of square domains.
We study two asymptotic limits - the $\lambda \to 0$ limit of vanishing cross-section size, and the $\lambda \to \infty$ limit relevant for larger micron-scale systems. For small $\lambda\to 0$, we have unique ring-like LdG minima which are well approximated by the Ring Solution analyzed in Propositions ~\ref{symmetry_proposition} and \ref{isotropic_center}. The Ring Solution, $\P_R$, has some generic properties for all polygons, $E_K$ with $K\geq 3$. For $K \neq 4$, $\P_R$ has a unique zero at the polygon centre which manifests as a uniaxial point with negative order parameter for the full $\Q$-tensor given by
$$\Q = \P_R - \frac{B}{6C}\left(2 \z\otimes \z - \x\otimes \x - \y\otimes \y \right). $$
We call this critical point a ``Ring'' solution since the unique zero has the profile of a degree $+1$-Ginzburg Landau vortex for $K> 4$. The case $K=4$ is special since the corresponding $\P_R$ vanishes along the square diagonals yielding an interesting cross pattern \cite{canevari2017order}.
For an equilateral triangle, the unique zero has the profile of a $-1/2$-nematic point defect as opposed to a unit vortex. Further differences arise if we work with irregular polygons e.g. an isosceles triangle as opposed to an equilateral triangle. We retain a unique zero for $\P_R$ but the location of the zero strongly depends on the angles between successive edges for isosceles triangles. In other words, we can manipulate the geometry of a polygon to control the nature of zeroes, the dimensions of the nodal set and their locations and this gives new vistas for control of equilibria, at least in the $\lambda \to 0$ limit. Ring-like solutions exist for all $\lambda$ and lose stability as $\lambda$ increases.

In the $\lambda \to \infty$ limit, we present a simple estimate for the number of stable reduced LdG equilibria accompanied by numerical results for a pentagon and hexagon.  In the case of polygons with an even number of $K$ sides, we always have at least $K/2$ classes of equilibria dictated by the locations of the ``splay" vertices and the number of vertices separating the ``splay" vertices. In the case of $E_6$, there are three families - Para, Meta and Ortho of which Para have the lowest energy (since the corresponding splay vertices are the furthest) and Ortho have the highest energy, with two neighbouring splay vertices.
Additionally, we have a class of BD solutions with two defective lines in the hexagon interior, which are connected to the Meta solution branches. The Ortho solution branches appear to be isolated.
For a pentagon, or more generally for a polygon with an odd number of $K$ sides, we expect to have $(K-1)/2$ families of stable equilibria dictated by the locations of the splay vertices. For $E_5$, there is no Para family and the BD solutions exist as unstable solution branches for all $\lambda$. Further, the BD solutions only have one defective line of ``low order'' for $E_5$. Whilst BD solutions are unstable, they are special since our numerics suggest that  they are index 1 saddle points with precisely one unstable direction.
 We have the numerical tools to compute the unstable directions and the indices of saddle points of the LdG energy \cite{yin2019high}. This would naturally lead to challenging problems in control theory if we want to control instabilities for applications, and cutting-edge questions in Morse theory, topology and integrability since the study of reduced LdG equilibria has intrinsic connections to entire solutions of certain integrable PDEs e.g. nonlinear sigma model, Allen-Cahn equation. Further, the methods in our paper also apply, to some extent, to the study of nematic equilibria in domains with inclusions or obstacles, where the nematic is in the exterior of a polygonal inclusion. For example, the authors study nematic equilibria outside a square obstacle with homeotropic anchoring in \cite{phillips2011texture}. They report stable string textures which resemble the WORS ($P_R$ on $E_4$), surface defect textures which resemble the rotated solutions in \cite{lewis2014colloidal} and stable textures with surface and bulk defects. We hope to pursue the generic similarities and differences between nematic equilibria in the interior and exterior of polygonal domains, including studies of saddle-point solutions, in future work.

\textbf{Acknowledgements.}
We would like to thank  Mr. Lidong Fang and Ms. Lingling Zhao for helpful discussions. Yucen Han also thanks the University of Bath and Keble College for their hospitality. 
Apala Majumdar thanks the University of Bath for a Visiting Professorship and the University of Oxford for an OCIAM Visiting Professorship, and the Royal Society for support from the Newton Advanced Fellowship. Apala Majumdar also thanks Samo Kralj for insightful discussions in 2014.
This research is supported by a Royal Society Newton Advanced Fellowship awarded to Professor Lei Zhang and Professor Apala Majumdar.
\bibliographystyle{unsrt}
\bibliography{han_SIAP_2020_arxiv}
\end{document}